\documentclass[showpacs,amsmath,amssymb,notitlepage,hidelinks]{revtex4-1}

\usepackage{color}
\usepackage{graphicx}
\usepackage{amsmath}
\usepackage{amsfonts}
\usepackage{amsthm}
\usepackage{amscd}
\usepackage{epsfig}
\usepackage{amssymb}
\usepackage{tabularx}
\usepackage{calligra}
\usepackage{enumerate}
\usepackage{hyperref}

\usepackage{amsthm}
\usepackage{verbatim} 

\newtheorem{thm}{Theorem}

\newtheorem{lem}{Lemma}

\renewcommand{\a}{\alpha}
\renewcommand{\b}{\beta}
\newcommand{\g}{\gamma}
\newcommand{\e}{\epsilon}

\renewcommand{\d}{\delta}

\newcommand{\w}{\omega}

\newcommand{\p}{\partial}

\newcommand{\M}{{\cal M}}

\def\td{\widetilde{D}}
\def\hd{\widehat{D}}

\begin{document}

\title{Black hole nonmodal linear stability: the Schwarzschild (A)dS cases}
\author{Gustavo Dotti} \email{gdotti@famaf.unc.edu.ar}
\affiliation{Facultad de Matem\'atica,
Astronom\'{\i}a y F\'{\i}sica, Universidad Nacional de C\'ordoba,\\
Instituto de F\'{\i}sica Enrique Gaviola, Conicet.\\
Ciudad Universitaria, (5000) C\'ordoba, Argentina}

\begin{abstract}
The nonmodal linear stability of the Schwarzschild black hole established in 
Phys.\ Rev.\ Lett.\  {\bf 112} (2014) 191101 is generalized to the case of a nonnegative 
cosmological constant $\Lambda$. 
Two gauge invariant 
combinations $G_{\pm}$ of perturbed scalars made out  of the Weyl tensor and its first covariant 
derivative are found such that the map $[h_{\alpha \beta}] \to \left( G_- \left([h_{\alpha \beta}]  \right),  G_+ \left([h_{\alpha \beta}]  \right) \right)$
with domain the set of equivalent classes  $[h_{\alpha \beta}]$  under gauge transformations of  solutions of the linearized Einstein's equation, is invertible. 
The way  to reconstruct a representative of $[h_{\alpha \beta}]$ in terms of $(G_-,G_+)$ is given.
It is proved that,  for  an arbitrary perturbation consistent 
with the background asymptote,  $G_+$ and $G_-$   are bounded in the the outer static region. At large times, 
 the perturbation decays leaving  a linearized Kerr black hole around the Schwarzschild or Schwarschild de Sitter background  solution. 
For negative cosmological constant it is shown that there are choices of boundary conditions at the time-like boundary 
under which the Schwarzschild anti de Sitter black hole is unstable. 
The root of Chandrasekhar's duality relating odd and even modes is exhibited, and some technicalities   related to this duality and 
 omitted in the original proof of the $\Lambda=0$ case are explained in detail.
\end{abstract}

\maketitle

\tableofcontents

\section{Introduction}
One of the most salient open problems in classical General Relativity (GR) is  proving the stability of the 
outer region of the 
stationary electro-vacuum black holes in the Kerr-Newman family. A complete proof of stability 
in the context of  the non linear GR equations  has only been given  for  Minkowski spacetime \cite{Christodoulou:1993uv};  
the stability problem of more complex solutions of Einstein's  equation is usually  approached by analyzing the behavior 
of linear 
 test fields satisfying appropriate boundary conditions in order to establish if  unbounded solutions are allowed.
Scalar test fields provide a first clue, whereas gravitational waves, that is, 
 metric perturbations $h_{\a \b}$ propagating on the background spacetime $(\M, g_{\a \b})$, 
give a more realistic  approach to the problem. For vacuum spacetimes, assuming a cosmological constant $\Lambda$, 
the metric perturbation $h_{\a \b}$ satisfies the linearized Einstein equation (LEE) 
\begin{equation} \label{lee}
{\cal E}[ h_{\a \b}] := - \tfrac{1}{2} \nabla^{\g} \nabla_{\g} h_{\a \b}  - \tfrac{1}{2} \nabla_{\a} \nabla_{\b} (g^{\g \d} h_{\g \d}) +
  \nabla^{\g} \nabla_{(\a} h_{\b) \g} - \Lambda \;  h_{\a \b} =0,
\end{equation}
obtained by assuming that on a fixed four dimensional manifold $\M$ (the spacetime) there is a 
 mono-parametric family of solutions $g_{\a \b}(\e)$ of the vacuum Einstein field equation
\begin{equation} \label{mpfs}
R_{\a \b}[g(\e)] - \Lambda g_{\a \b}(\e)=0
\end{equation}
around the ``unperturbed background" $g_{\a \b}=g_{\a \b}(0)$, 
and Taylor expanding (\ref{mpfs}) at $\e=0$. For  tensor fields that depend 
functionally on the metric we use a  dot to denote the ``perturbed field'', 
which is the  field obtained by taking  the first derivative with respect to $\e$  at $\e=0$, e.g.,
for the Riemann tensor, 
\begin{equation}
\dot R^{\a}{}_{\b \g \d} = \frac{d}{d \e}\bigg|_{\e=0} R^{\a}{}_{\b \g \d}[g(\e)].
\end{equation}
We   make an  exception for the metric field itself by adopting  the standard notation  $h_{\a \b} = \dot g_{\a \b}$ and, as usually
 done  in linear perturbation theory, defining 
$h^{\a}{}_{\b} =g^{\a \mu} h_{\mu \b}$, and $h^{\a \b} = g^{\a \mu} g^{\b \nu} h_{\mu \nu}$. Note that this  convention implies that 
$$ \dot g^{\a \b} = - h^{\a \b}.$$
Equation (\ref{lee}) is the first order Taylor coefficient  of (\ref{mpfs}) around $\e=0$, that is 
${\cal E}[ h_{\a \b}]  = \dot R_{\a \b} - \Lambda h_{\a \b}$. 
Trivial solutions of this equation  are 
\begin{equation} \label{puregauge}
 h_{\a \b} =  \pounds_{\xi}   g_{\a \b} = \nabla_{\alpha} \xi_{\beta} +
\nabla_{\beta} \xi_{\alpha},
\end{equation}
where $\xi^{\a}$ is an arbitrary  vector field; these amount to 
the first order in $\e$ change of the metric under the pullback by the flow $\Phi_{\e}^{\xi}: \M \to \M$ generated by  the vector field $\xi^{\a}$. 
 Any two
 solutions $h_{\a \b}$ and $h'_{\a \b}$ such that 
\begin{equation} \label{gt2}
h'_{\a \b} = h_{\a \b} +  \pounds_{\xi}   g_{\a \b},
\end{equation}
 are related by this diffeomorphism and thus physically equivalent, this being the {\em gauge freedom} of linearized gravity. 
If a tensor field $T$ is a functional of the metric $g$, then $\dot T$ is a linear functional of $h_{\a \b}$ and, 
under the {\em gauge transformation} (\ref{gt2}) we find that 
\begin{equation} \label{gtt}
\dot T[h'] = \dot T[h] + \pounds_{\xi}  T.
\end{equation}
In particular, if $T$ is a scalar field, $\dot T$ is gauge invariant iff $0=\pounds_{\xi}  T = \xi^{\a} \p_{\a} T$ for every 
vector field, that is, 
if $T$ is constant in the background.\\
 Curvature related scalar fields (CS for short, not to be confused with the  tetrad components of the Weyl tensor in the 
Newman-Penrose formalism), are  scalar fields obtained 
by a full contraction of a polynomial
in the Riemann tensor and  its 
covariant derivatives, the metric and the volume form. Although these fields   partially characterize the metric, it is well known  that 
the information they carry is limited, an extreme example of this fact  being the  {\em vanishing scalar invariants spacetimes},  which 
have a nonzero Riemann tensor and yet {\em every} CS vanishes \cite{Pravda:2002us}. 
This fact may suggest that the perturbation of CS
 under a given solution $h_{\a \b}$ of the LEE  provide very limited information on $h_{\a \b}$. 
It was shown in  \cite{Dotti:2013uxa}, however, that for a Schwarzschild black hole background, 
 there are two gauge invariant combinations $G_{\pm}$ of perturbed  CSs that fully 
characterize the gauge equivalence class $[h_{\a \b}]$ of the corresponding  solution of the LEE. 
More precisely, let $C^{\a}{}_{\b \g \d}$ be the Weyl tensor,  $\varepsilon_{\a \b \mu \nu}$ the volume form, and consider  the CSs
\begin{equation} \label{curvaturescalars}
Q_- = \tfrac{1}{96} C^{\a \b \g \d} \varepsilon_{\a \b \mu \nu}  C^{\mu \nu}{}_{\g \d}, \;\;\;\;\;
 Q_+ = \tfrac{1}{48} C^{\a \b \g \d}   C_{\a \b \g \d} , \;\;\;\;\;
X = \tfrac{1}{720}  \left( \nabla_{\e} C_{\a \b \g \d} \right)  \left( \nabla^{\e} C^{\a \b \g \d} \right).
\end{equation}
The background values of these fields, that is, their value for the $\Lambda=0$ Schwarzschild solution, are 
${Q_-}_{Schw}^0 = 0$, ${Q_+}_{Schw}^0 = M^2/r^6$ and $X_{Schw}^0 = M^2 (r-2M)/r^9$, where $M$ is the mass and $r$ the areal radius, 
then the fields 
\begin{equation} \label{gpm}
G_- = \dot Q_-, \;\;\;\;  G_+=(9M-4r) \dot Q_+ + 3 r^3 \dot X, 
\end{equation}
made out of the first order variations of these CS, are gauge invariant.  It was shown in  \cite{Dotti:2013uxa} 
that the linear map
\begin{equation} \label{im}
[h_{\a \b}] \to \left( G_{-}\left([h_{\a \b}] \right),  G_{+}\left([h_{\a \b}] \right) \right)
\end{equation}
with domain the equivalence classes of smooth solutions of the LEE  preserving the asymptotic flatness,  is injective. 
This implies 
 that the scalars $G_{\pm}$, beyond  measuring the ``amount of distortion'' caused by the perturbation, 
encode {\em all} the relevant information on the perturbation class $[h_{\a \b}]$. A way to construct a class representative $h_{\a \b}$ from 
the fields $G_{\pm}\left([h_{\a \b}] \right)$ is  indicated in \cite{Dotti:2013uxa}. \\

The stability concept introduced in \cite{Dotti:2013uxa} is based on i) the existence of the CSs $G_{\pm}$ 
for which (\ref{im}) is injective on smooth solutions of the LEE 
that preserve the asymptotic flatness, and  ii) the proof that for this class of solutions of the LEE 
the scalar fields $G_{\pm}\left([h_{\a \b}] \right)$ are bounded. More precisely, it is proved in \cite{Dotti:2013uxa} 
that in the outer region $r \geq 2M$  
\begin{equation} \label{bounds1}
|G_-| < \frac{K_-}{r^6}, \;\;\;\; |G_+| < \frac{K_+}{r^3},
\end{equation}
where  $K_{\pm}$ are constants that can be obtained from  the  perturbation field datum at a Cauchy slice.  Given that 
the scalars $G_{\pm}$ contain {\em all} the gauge invariant information on the perturbation, the fact that they remain bounded 
as the perturbation evolves through   spacetime gives  a meaningful notion of linear stability.\\

This concept of stability, that we call {\em nonmodal},  should be compared 
 with prior  linear gravity stability notions for the Schwarzschild black hole. To this end we review 
the results in 
a short list  of papers that were crucial in the development  of this subject.  
It is important to stress  that they all  use the spherical symmetry of the Schwarzschild spacetime 
to decompose a metric perturbation 
\begin{equation} \label{dec}
h_{\a \b}= \sum_{(\ell,m,p=\pm)} h_{\a \b}^{(\ell,m,p=\pm)}
\end{equation}
 into even ($p=+$) and odd ($p=-$)
 $(\ell,m)$ {\em modes}. Here $\ell$ refers to 
the  eigenspace of the Laplace-Beltrami operator acting on real scalar fields on  $S^2$ corresponding to the eigenvalue $-\ell(\ell+1)$, $m$ is an index labeling a particular 
basis of this  $2\ell +1$ dimensional space, and the parity $p$  accounts for the way $h_{\a \b}^{(\ell,m,p=\pm)}$ transforms  when pull-backed 
by the antipodal map on $S^2$ (for details refer to Section \ref{cstf}). \\
The first work on the linear stability of the Schwarzschild black hole is T. Regge and J. Wheeler 1957 paper \cite{Regge:1957td}, where  the decomposition 
(\ref{dec})  was proposed and the $\ell=0,1$ modes where recognized to be non-dynamical. At the time the very notion 
of black hole was unclear (the term ``black hole'' was coined by J. Wheeler some ten years later),
 and Kerr's solution had not yet been discovered. This is probably why,  although 
the $\ell=0$ even piece of the perturbation was readily identified as a mass shift in \cite{Regge:1957td}, the odd $\ell=1$ modes, which 
corresponds to perturbing along a Kerr family $g_{\a \b}(\e)$  with  $\e=J/M$, was misunderstood (see the paragraph 
between equations (37) and (38) in \cite{Regge:1957td}) and the opportunity of producing a ``slowly rotating'' black hole at a time  when there was no clue about 
a rotating black hole solution was  missed 
\footnote{I thank Reinaldo Gleiser for this observation.}. The decomposition (\ref{dec}) in \cite{Regge:1957td} was done in Schwarzschild coordinates 
$(t,r,\theta,\phi)$ in 
what 
came to be called the Regge-Wheeler 
(RW, for short) gauge. The LEE in the even and odd sectors were shown to decouple, and the dynamical odd sector of the 
LEE  reduced 
to an infinite set of two dimensional wave equations, individually know as  the RW equation:
\begin{equation} \label{RWE}
(\p_t^2 - \p_{r^*}^2 + f V_{(\ell,m)}^{RW} ) \phi_{(\ell,m)}^- =0,  \;\;\;\; \ell \geq 2, \tag{RWE}
\end{equation}
where (adding a cosmological constant $\Lambda$ for future reference)  $\p_{r^*} = f \p_r$, $f=1-2M/r-\Lambda r^2/3$, and 
\begin{equation} \label{rwp}
V_{(\ell,m)}^{RW}= \frac{\ell(\ell +1)}{r^2} - \frac{6M}{r^3}.
\end{equation}
The $\ell \geq 2$ even sector LEE, a much more intricate system of equations, was disentangled by F. Zerilli in his 1970 paper 
\cite{Zerilli:1970se} and shown to be equivalent to the wave equations
\begin{equation} \label{ZE} \tag{ZE}
(\p_t^2 - \p_{r^*}^2 + f V_{(\ell,m)}^{Z} ) \phi_{(\ell,m)}^+ =0,   \;\;\;\; \ell \geq 2, 
\end{equation}
with potentials 
\begin{equation} \label{zp}
V_{(\ell,m)}^Z= \frac{[ \mu^2 \ell (\ell+1)-24M^2 \Lambda] r^3 + 6 \mu^2 M r^2 +36 \mu M^2 r +72M^3}{r^3 \; 
(6M+ \mu r)^2}, \;\; \mu=(\ell-1)(\ell+2).
\end{equation}
For $\Lambda \neq 0$ the RW and Zerilli potentials first appeared  in \cite{guven}. 
The solution of the LEE in the RW gauge is then given in the form 
\begin{equation} \label{rws}
{}^{(RW)}h_{\a \b}^{(\ell,m,p=\pm)}=\mathcal{D}_{\a \b}^{(\ell,m,p=\pm)}
\left[\phi_{(\ell,m)}^{\pm},S_{(\ell,m)}\right]
\end{equation}
where $\mathcal{D}_{\a \b}^{(\ell,m,p=\pm)}$ is a bilinear differential operator acting on $\phi_{(\ell,m)}^{\pm}$ and the 
spherical harmonics $S_{(\ell,m)}$. Note that, since focus is on non-stationary modes, only $\ell \geq 2$ enter this formulation. 
The non-stationary solution space of the  LEE is thus parametrized by the infinite set of scalar fields $\phi_{(\ell,m)}^{\pm}$ 
that enter the series (\ref{dec}) through (\ref{rws}). \\

Every notion of  gravitational linear stability of the Schwarzschild black hole prior to \cite{Dotti:2013uxa} was concerned with the behavior 
of the potentials $\phi_{(\ell,m)}^{\pm}$ of isolated $(\ell,m)$ modes (we  call this  ``modal linear stability''). In particular: 
\begin{itemize}
\item In \cite{Regge:1957td} it was shown that  separable solutions $\phi_{(\ell,m)}^- = \Re\; e^{i \w t} \chi_{(\ell,m)}^-(r)$ 
that do not diverge as $r \to \infty$ require $\w \in \mathbb{R}$, ruling out exponentially growing solutions in the odd sector. 
\item In \cite{Zerilli:1970se} it was shown that  separable solutions $\phi_{(\ell,m)}^+ = \Re \; e^{i \w t} \chi_{(\ell,m)}^+(r)$ 
that do not diverge as $r \to \infty$ require $\w \in \mathbb{R}$, ruling out exponentially growing solutions in the even sector. 
\item In \cite{Price:1971fb} it was shown that, for large $t$ and  fixed $r$,   $\phi_{(\ell,m)}^{\pm}(t,r)$ decays as $t^{-(2\ell+2)}$ 
(an effect known as ``Price tails'').
\item In \cite{wald}, the conserved energy
\begin{equation}
\int_{2M}^{\infty} \left[(\p_t \phi_{(\ell,m)}^{\pm})^2 + (\p_{r^*} \phi_{(\ell,m)}^{\pm})^2 + \phi_{(\ell,m)}^{\pm} 
f V_{(\ell,m)}^{RW/Z} \phi_{(\ell,m)}^{\pm} 
\right] f dr
\end{equation}
was used to rule out uniform exponential growth in time.
\item Also in \cite{wald}, a point wise bound on the RW and Zerilli potentials was found in the form 
\begin{equation} \label{modal-b}
|\phi_{(\ell,m)}^{\pm}(t,r)| < C_{(\ell,m)}^{\pm}, \;\; r>2M, \text{ all } t,
\end{equation}
where the constants $C_{(\ell,m)}^{\pm}$ are given in terms of the initial data  
$(\phi_{(\ell,m)}^{\pm}(t_o,r), \p_t \phi_{(\ell,m)}^{\pm}(t_o,r))$
\end{itemize}
 To understand the limitations of these results it is important to keep in mind that 
  the $ \phi^{\pm}_{(\ell,m)}$ are  an infinite set of  potentials defined on the $(t,r)$ orbit space $\mathcal{M}/SO(3)$, whose   
first and second order derivatives  $ \phi^{\pm}_{(\ell,m)}$
 enter the terms in the series (\ref{dec})  through (\ref{rws}), together with first and second derivatives of 
the spherical harmonics. 
 Two extra derivatives are required  to calculate the perturbed Riemann tensor, as  a  first step  to measure  the 
effects of the perturbation on the curvature. Thus,   the relation of the  potentials $\phi_{(\ell,m)}^{\pm}$
 to geometrically meaningful
quantities is remote, and the usefulness 
  of the bounds (\ref{modal-b}) to measure the strength of the perturbation   is far from  obvious.  \\

The motivation of the nonmodal approach  can be better understood if we 
put in perspective the progress made in \cite{wald}, equation  (\ref{modal-b}). 
  The fact that the two dimensional wave equations (\ref{RWE}) and (\ref{ZE}) 
can be solved by separating variables ($\phi_{(\ell,m)}^{\pm} = e^{i \w t} \psi(r)$) had previous works on linear stability 
focus on 
showing that $\w$ must be real, a limited notion of stability that does not even 
rule out, e.g.,  linear  $t$ growth \cite{wald}.  
To obtain the type of bounds (\ref{modal-b}) one must ``undo'' the separation of variables and reconsider the original equations 
 (\ref{RWE}) and (\ref{ZE}),  as this allows  to work out  results that are  valid for 
 generic, non separable solutions. The idea behind the non modal stability concept introduced  in this paper,  
 is to go one step further and 
  ``undo'' the   separation of the $(\theta,\phi)$ variables that antecedes the reduction of   the LEE 
to  (\ref{RWE}) and (\ref{ZE}),  in order to get bounds for truly four dimensional quantities. 
When trying to make  this idea more precise, 
one is  immediately faced to the problem of {\em which} four dimensional quantity one should  look at. 
If this quantity is to measure the strength of the perturbation, it should be a field related to the curvature change,
 and it should also 
be a {\em scalar} field, since 
there is no natural norm for a tensor field in a Lorentzian  background. If this scalar field does not obey 
a four dimensional wave equation (or some  partial differential equation), we do not have a clear 
tool to place bounds on it. Thus, we are naturally led to the question of the existence of scalar fields made from 
perturbed CSs that,  as a consequence of the LEE, obey some wave-like equation. This is the problem we address in this paper. 
The name  {\em nonmodal stability} is borrowed from fluid mechanics, where the limitations of the normal mode analysis 
were realized some thirty years ago  in 
 experiments involving wall bounded shear flows  \cite{schmid}. In that problem, the 
linearized Navier-Stokes operator is non normal, so their eigenfunctions are non orthogonal and, as a consequence, even 
if they individually decay as $e^{i w t}, \Im (w) >0$, a condition that assures large $t$ stability,  there may be important transient 
growths \cite{schmid}. Take, e.g., the following simple example (from \cite{schmid}, section 2.3 and figure 2)  
of a system of two degrees of freedoms: $\vec{v} \in 
\mathbb{R}^2$ obeying the equation $d \vec{v}/dt = -A \vec{v}$, with $[A,A^T] \neq 0$ a matrix with (non orthogonal!) eigenvectors 
$\vec{e}_1, \vec{e}_2$, say, 
\begin{equation}
A = \left( \begin{array}{cc} 1 & \g \\ 0 & 2 \end{array} \right), \;\;\;  
\vec{e}_1= \left( \begin{array}{c} 1 \\ 0 \end{array} \right),  \;\;\;  
\vec{e}_2= \left( \begin{array}{c} 1 \\ \g^{-1} \end{array} \right)
\end{equation}
Consider the case $\g >>1$. Note that, although $\w_1=i, \w_2=2i$, that is, normal modes decay exponentially,  if 
 $\vec{v}(0)=\a (\vec{e}_1-\vec{e}_2)$, then $\vec{v}(t)=\a \vec{e}_1 e^{-t}- \a \vec{e}_2 e^{-2t}$ 
reaches a maximum norm $||\vec{v}|| \simeq (\g/4) ||\vec{v}(0)||$ at a finite time before decaying to zero. \\

For the odd sector of the LEE, a four dimensional approach relating the metric perturbation  with 
a scalar potential  $\Phi$ defined on $\mathcal{M}$ (instead of $\mathcal{M}/SO(3)$) was 
found in \cite{Dotti:2013uxa}, where it was noticed that the sum over $(\ell,m)$ of (\ref{rws}) simplifies to 
\begin{equation} \label{CsP}
{}^{RW}h_{\a \b}^-= \sum_{(\ell \geq 2,m)} \mathcal{D}_{\a \b}^{(\ell,m,-)}
\left[\phi_{(\ell,m)}^{-},S_{(\ell,m)}\right] = \frac{r^2}{3M} {}^* C_{\a}{}^{\g \d}{}_{\b} \nabla_{\g} \nabla_{\d}
 \left( r^3 \Phi \right),
\end{equation}
where $\Phi$ is a field assembled using spherical harmonics and the RW potentials:
\begin{equation} \label{4DRf}
 \Phi = \sum_{(\ell \geq 2,m)} \frac{\phi^-_{(\ell,m)}}{r} S_{(\ell,m)} : \mathcal{M} \to \mathbb{R}. 
\end{equation}
The odd sector LEE equations  (\ref{RWE}) for $\phi^-_{(\ell,m)}:  \mathcal{M}/SO(3) \to \mathbb{R}$ 
combined to the spherical harmonic equations for the spherical harmonics $S_{(\ell,m)}: S^2 \to \mathbb{R}$,  turn out 
to be equivalent to what we call the four dimensional Regge-Wheeler equation which, adding a cosmological constant,  reads
\begin{equation} \tag{4DRWE} \label{4DRWE}
\nabla^{\a} \nabla_{\a} \Phi + \left( \frac{8M}{r^3} - \frac{2 \Lambda}{3}  \right) \Phi =0. 
\end{equation}
Note however that  $\Phi$ is no more than the collection of $\phi^-_{(\ell,m)}$'s, so its connection to geometrically relevant 
fields is loose.   
Much more important is the fact, also proved in \cite{Dotti:2013uxa} for $\Lambda=0$, that 
 the LEE implies that   the field $r^5 \dot Q_- = 
r^5 G_-$ {\em also} satisfies the \ref{4DRWE}, as  this is what allows us   to place a point wise bound on $G_- = \dot Q_-$. \\
The even sector of the LEE is more difficult to approach. Is the simplicity of the RW potential (\ref{rwp}), 
with the obvious  $\ell (\ell+1)/r^2$ term, what suggested considering the field (\ref{4DRf}). The way  $\ell$ appears in (\ref{zp}), instead,  
is a clear indication that there is no natural 4D interpretation of (\ref{ZE}). Is a map exchanging solutions of  the \ref{RWE} and 
\ref{ZE} equations, found by Chandrasekhar \cite{ch2},  what ultimately allows us to also reduce the even non
 stationary   LEE equations to (\ref{4DRWE}). As a consequence,
 the entire set of non stationary  LEE reduces    to two fields satisfying equation 
(\ref{4DRWE}), as stated in Theorem \ref{bij} below.  \\

The purpose of  this paper is twofold: (i) to  extend the results in \cite{Dotti:2013uxa} 
to  Schwarzschild  black holes in cosmological backgrounds, and (ii) to explain in detail 
a number of technicalities omitted in \cite{Dotti:2013uxa} due to the space limitations imposed by the letter format. 
For $\Lambda \geq 0$ we give a proof of non-modal stability. 
 We leave aside  the treatment of stability of the Schwarzschild anti de Sitter (SAdS) black hole, since the issue of non global hiperbolicity and 
ambiguous dynamics due to the conformal timelike boundary takes us away from 
the core of the subjects addressed  here.
We show however that there is (at least) one choice of 
 Robin boundary condition at the time-like boundary for which 
there is an instability, and we explicitly  exhibit this instability and its effect on the background geometry. 
 To the best of our knowledge, this has not been informed before. 
A systematic study of the  gravitational linear stability of SAdS black holes under different   boundary conditions 
is to be found in \cite{bernardo}. \\

We have found that  $(G_+,G_-)$ defined in (\ref{gpm}) 
are appropriate variables to study the most general gravitational linear perturbations of  Schwarzschild (A)dS black holes. 
The $\ell=0,1$ pieces of these fields encode the relevant information on the stationary modes, 
which are perturbations within the Kerr family (parametrized by  mass $M$ and the 
angular momenta components),
whereas the $\ell \geq 2$ terms encode the dynamics. More precisely:
\begin{itemize}
\item $G_-$ contains no $\ell=0$ term, 
time independent $\ell=1$ terms proportional to the first order angular momenta components ${j^{(i)}}$, and a time dependent $\ell>1$ piece obeying  
the 4DRW equation.
\item $G_+$ contains no $\ell=1$ term, a time independent $\ell=0$ piece proportional to 
a mass shift $\dot M$, and a time dependent $\ell>1$ piece which, for $\Lambda \geq 0$,  can be written 
 in terms of  fields obeying the 4DRW equation. 
\end{itemize}
Once the appropriate set of perturbation fields $(G_+,G_-)$ is given, and their relation to the 4DRW equation  established for $\Lambda \geq 0$, 
we may adapt to the 4DRW equation  the techniques used to prove  boundedness of solutions of the scalar wave equation,  
in order to analyze the behavior of the $G_{\pm}$ fields. As an example, the  result of Kay and Wald \cite{Kay:1987ax} was 
used  in  \cite{Dotti:2013uxa} to prove (\ref{bounds1}) in the Schwarzschild case, and is adapted here to prove  that  (\ref{bounds1}) holds 
also for positive $\Lambda$. 
 We can go further and 
 take advantage of the growing literature on decay of solutions of the scalar wave equation  on S(A)dS backgrounds, as many of these results 
are expected  to hold also for  (\ref{4DRWE}). Specific time decay results for the 4DRW equation, somewhat expected 
from Price's result \cite{Price:1971fb}, can be found in \cite{Blue:2003si} (see also the recent preprint \cite{Dafermos:2016uzj}). 
Putting together  the bijection (\ref{im}), the above description of the stationary ($\ell=0,1$) and dynamic ($\ell \geq 2$) 
pieces of $G_{\pm}$, and the time decay results, 
the following picture emerges for a perturbed Schwarzschild (SdS) black hole: a generic perturbation contains a mass shift, infinitesimal angular momenta 
$j^{(i)}$ and  dynamical degrees of freedom; at large times the dynamical degrees of freedom decay and what is left 
is a linearized    Kerr (Kerr dS) black hole around the background Schwarzschild (SdS) solution. \\

Through the paper, calculations are carried leaving $\Lambda$ unspecified whenever possible, and specializing when necessary. 
Among the many current treatments of 
linear perturbations of spherically symmetric spacetimes, we have made heavy  use of the 
excellent paper \cite{Chaverra:2012bh}, which we found particularly well suited 
to our approach.

\section{Tensor fields on a spherically symmetric space-time}

A spherically symmetric space-time is a warped product ${\cal M} = {\cal O} \times_{r^2} S^2$ of a Lorentzian 
two-manifold $(\mathcal{O}, \tilde g)$ with the unit sphere $(S^2, \widehat g)$, for which we will use the   standard angular coordinates 
 $ \widehat g_{AB} dy^A dy^B = d \theta^2 + \sin^2 \theta \; d\phi^2$:
\begin{equation} \label{ssm}
g_{\a \b} dz^{\a} dz^{\b} = \tilde g_{ab}(x) dx^a dx^b + r^2(x)  \widehat g_{AB}(y) dy^A dy^B. 
\end{equation}
The form of the metric   (\ref{ssm})  implies that  $(\mathcal{M},g)$ inherits the isometry group $O(3)=SO(3) \times P$ of $(S^2,\widehat g)$ 
as an isometry subgroup. Here  
 $SO(3)$ are the proper rotations and $P$ is the antipodal map $P(\theta,\phi)=(\pi-\theta,\phi+\pi)$.
Since $SO(3)$ acts transitively on $S^2$, we find that ${\cal O} = {\cal M} / SO(3)$, this is  why $(\mathcal{O}, \tilde g)$
 is  called the ($SO(3)$) 
{\em orbit space}.\\

Equation (\ref{ssm}) exhibits our index conventions, which we  have adopted form ref \cite{Chaverra:2012bh}:
  lower case  Latin indexes are used  for tensors on  ${\cal O}$, upper case  Latin indexes 
for tensors on $S^2$, and Greek indexes  for space-time tensors. We will furthermore assume that 
\begin{equation} \label{index}
\a=(a,A),\; \b=(b,B), \; \g=(c,C), \; \d=(d,D),... 
\end{equation}
Tensor fields {\em   introduced}
 with a lower $S^2$ index (say $Z_A$) and {\em then shown with an upper $S^2$ index} are assumed 
to have been acted  upon {\em with the unit $S^2$ metric inverse $ \widehat g^{AB}$},
 (i.e.,  $Z^A \equiv  \widehat g^{AB} Z_B$), 
and similarly for upper $S^2$ indexes moving down.  
$\td_a, \tilde \e_{ab}$ and $\tilde g^{a b}$ are the covariant derivative, volume form (any chosen orientation) and metric inverse 
for  $(\mathcal{O}, \tilde g)$; $\hd_A$ and  $ \widehat \e_{AB}$ are the covariant derivative and volume form 
on the unit sphere, for which we assume  the standard orientation $\widehat \e = \sin (\theta) d\theta \wedge d\phi$. \\
As an example,  in terms of the differential operators $\td_a$ and $\hd_A$,  the Laplacian on 
scalar fields reads
\begin{equation} \label{lap}
\nabla_{\a} \nabla^{\a} \Phi = \td_a \td^a \Phi + \frac{2}{r} (\td^b r) (\td_b \Phi) + \frac{1}{r^2} \hd_A \hd^A \Phi.
\end{equation}

\subsection{Covector and symmetric tensor fields} \label{cstf}

The Einstein field equation, as well as its linearized version around a particular solution, is expressed as an equality among symmetric tensor fields. 
The metric perturbation $h_{\a \b}$ may be subjected to gauge transformations of the form (\ref{gt2}). 
This is why we are interested  in the decomposition on a spherically symmetric spacetime of covector fields 
(such as $\zeta_{\a}$) 
and symmetric rank two tensor fields such as $h_{\a \b}$. \\
We will assume all tensor fields on $\mathcal{M}$ are smooth. As a consequence their components will be square integrable on $S^2$ and 
can be expanded using a real orthonormal basis of scalar spherical harmonics 
\begin{align}
& \hd^A \hd_A S_{(\ell, m)} + \ell (\ell+1) S_{(\ell, m)} =0, \\ 
&\int_{S^2}  S_{(\ell', m)}  S_{(\ell, m)} \widehat \e = \d_{\ell  \ell'} \d_{m m'},
\end{align}
where $m$ numbers an orthonormal  basis $S_{(\ell,m)}$ of the $(2\ell +1)$-dimensional eigenspace with eigenvalue $-\ell(\ell+1)$ 
of the Laplace-Beltrami operator $\hd^A \hd_A$ on scalar functions ($\ell=0,1,2,3,...$). An explicit choice for the $\ell=0,1$ subspaces is 
$
S_{(0, 0)}  = \sqrt{\frac{1}{4\pi}},  
$
and
\begin{align} 
\begin{split} \label{l=1}
S_{(1,1)} &= \sqrt{\frac{3}{4\pi}} \sin \theta \cos \phi \\
S_{(1,2)} &= \sqrt{\frac{3}{4\pi}} \sin \theta \sin \phi \\
S_{(1, 3)}  &= \sqrt{\frac{3}{4\pi}} \cos \theta.
\end{split}
\end{align}
We denote $L^2(S^2)_{\ell}$   the $\ell$ subspace of $L^2(S^2)$ and  
$L^2(S^2)_{>j} = \bigoplus_{\ell>j}  L^2(S^2)_{\ell}$.\\

A covector field on $\mathcal{M}$
\begin{equation} \label{vd}
\xi_{\a}=(\xi_a, \xi_A),
\end{equation}
contains the $S^2$ covector $\xi_A$ which, according to Proposition 2.1 in \cite{Ishibashi:2004wx}, 
can be  uniquely decomposed as $\xi_{A}=\hd_A a + V_A$ where $\hd^A V_A=0$. This last condition implies that 
$V_A$ is dual to the differential of  an $S^2-$scalar, $V_A =   \widehat \e_{A}{}^C \hd_C b$. It then follows that, 
introducing  $a=r^2 X$ and $b= r^2 Y$  for later convenience, 
\begin{equation} \label{xidecomp}
\xi_{\a}=   \xi_{\a}^{(-)}  + \xi_{\a}^{(+)} , 
\end{equation}
where the {\em odd} piece of $\xi_{\a}$ is 
\begin{equation} \label{vec-}
\xi_{\a}^{(-)} = (0, r^2 \; \widehat \e_{A}{}^C \hd_C Y),
\end{equation}
and its {\em even} piece is  
\begin{equation} \label{vec+}
\xi_{\a}^{(+)} = (\xi_a, r^2\; \hd_A X). 
\end{equation}
For a given covector field $\xi_{\a}$, the scalar fields $X, Y: \mathcal{M} \to \mathbb{R}$ are unique up to an $S^2$-constant, thus  they are  unique if we require  that 
they belong to $L^2(S^2)_{>0}$
\begin{equation} \label{0l0}
\int_{S^2} X\;  \widehat \e = 0 = \int_{S^2} Y \; \widehat \e,
\end{equation} 
a condition that we will assume. 
The  symmetric tensor field $h_{\a \b} = h_{\b \a}$ 
\begin{equation} \label{std1}
h_{\a \b} = \left( \begin{array}{cc} h_{ab} & h_{aB} \\ h_{Ab} & h_{AB} \end{array} \right),
\end{equation}   
contains two $S^2$ covector fields 
\begin{equation} \label{std2}
h_{aB} = \hd_B q_a +  \widehat \e_{B}{}^C \hd_C h_a 
\end{equation}
and a symmetric $S^2$ tensor field $h_{AB}$. Note that 
$q_a$ and $h_a$ are covector fields on $\mathcal{O}$ parametrized on $S^2$.\\
Using Proposition 2.2 in \cite{Ishibashi:2004wx} and the fact that 
there are no transverse traceless symmetric tensor fields on $S^2$ \cite{Higuchi:1986wu}, we find that the $S^2$ symmetric tensor $h_{AB}$ in (\ref{std1}) 
 can be uniquely decomposed into  three terms:
\begin{equation} \label{std3}
h_{AB} = 2 \hd_{(A} (\widehat \e_{B) C} \hd^C L) + \left(2 \hd_A \hd_B -  \widehat g_{AB} \hd^C \hd_C \right) S_1 +  \widehat g_{AB} \frac{S_2}{2}.
\end{equation}
Introducing $J=S_2/r^2$, $G=S_1/r^2$ and $k=L/r^2$ we arrive at (c.f.,  \cite{Chaverra:2012bh}, Section IV.A)\\

\begin{lem} A generic smooth metric perturbation admits the following decomposition:\\
\begin{equation} \label{pert}
h_{\a \b} = h_{\a \b}^{(-)} + h_{\a \b}^{(+)}
\end{equation}
where the {\em odd} piece of $h_{\a \b} $ is 
\begin{equation} \label{pert-}
 h_{\a \b}^{(-)} =  \left( \begin{array}{cc} 0 &  \widehat \e _B{}^C \hd_C h_a \\  \widehat \e _A{}^C  \hd_C h_b & 2 r^2\widehat \e_{(A}{}^C  \hd_{B)}  \hd_C k 
 \end{array} \right)
\end{equation}
and the {\em even} piece is 
\begin{equation} \label{pert+}
 h_{\a \b}^{(+)} = \left( \begin{array}{cc}  
h_{ab} & \hd_B q_a \\ \hd_A q_b & r^2 \left[\tfrac{J}{2} \widehat g_{AB} + 
(2 \hd_A \hd_B - \widehat g_{AB}\hd^C \hd_C) G \right] \end{array} \right).
\end{equation}
\end{lem}

The proof of the following lemma follows from  straightforward calculations:

\begin{lem} \label{kernel}
 \mbox{}
\begin{itemize}
\item[(i)] 
The kernel of the map $(h_a, k) \to h_{\a \b}^{(-)}$ defined in (\ref{pert-}) is the set of $h_a$ and $k$ of the form
\begin{equation}
h_a = h_a(x), \;\; k = k_1(x)+k_2(x) \cos(\theta) + \sin(\theta) \left( k_3(x)  \cos(\phi) + k_4(x)   \sin(\phi) \right),  
\end{equation}
This implies that $h_a$ and $k$ are unique  if they are required to belong to $L^2(S^2)_{>0}$ and 
$L^2(S^2)_{>1}$   respectively:
\begin{equation} \label{exp1}
h_a = \sum_{\ell \geq 1, m} h_a^{(\ell,m)}(x) S_{(\ell,m)}(\theta,\phi), \;\; k(x) = \sum_{\ell \geq 2, m} k^{(\ell,m)}(x) S_{(\ell,m)}(\theta,\phi).
\end{equation}
\item[(ii)] The kernel of 
the map $(h_{ab}, q_a, J, G) \to h_{\a \b}^{(+)}$ defined in (\ref{pert+}) is 
characterized by $h_{ab}=0, J=0,$
 \begin{equation}
q_a = q_a(x), \;\; G = G_1(x)+G_2(x) \cos(\theta) + \sin(\theta) \left( G_3(x)  \cos(\phi) + G_4(x)   \sin(\phi) \right),  
\end{equation}
thus, the fields   $(h_{ab}, q_a, J, G)$ are uniquely defined if we require that $q_a \in L^2(S^2)_{>0} $ and $G\in L^2(S^2)_{>1}$. 
\end{itemize}
\end{lem}

From now on we will assume the required conditions for uniqueness of $k, J, G, h_a, q_a$ and $h_{ab}$. \\

The linearized Ricci tensor $\dot R_{\a \b}$ admits a decomposition analogous to (\ref{std1})-(\ref{pert+}). 
Given that $S^2-$scalar fields, divergence free covector fields (which are all of the form $\widehat \e_A{}^B \hd_B C$) and transverse traceless symmetric tensors 
on $S^2$ span  inequivalent $O(3)$ representations,  and that the linear map 
$h_{\a \b} \to \dot R_{\a \b}$ is $O(3)$ invariant, this map cannot mix odd and even sectors \cite{Ishibashi:2004wx}. This implies that
$\dot R_{\a \b}^{(+)}$ is a linear functional of $ h_{\a \b}^{(+)} $ only, and similarly 
$\dot R_{\a \b}^{(-)}$ depends only on $ h_{\a \b}^{(-)}$.  \\

Note from (\ref{l=1}) that 
\begin{equation} \label{kvfs}
J^A_{(m)} = \sqrt{\tfrac{4\pi}{3}}\;\; \widehat \e^{AB} \hd_B S_{(1,m)},
\end{equation}
$ m=1, 2, 3$, is a basis of Killing vector fields on $S^2$ generating rotations around orthogonal axis, normalized such that 
the maximum length of their orbits is $2 \pi$.  The square angular momentum operator is the sum of the squares of 
the Lie derivatives along these  Killing vector fields:
\begin{equation}
\mathbf{J}^2 = \sum_{m=1}^3 (\pounds_{J_{(m)}})^2.
\end{equation}
This operator commutes with the maps $(h_a, k) \to h_{\a \b}^{(-)}$ and $h_{\a \b}^{(-)} \to R_{\a \b}^{(-)}$ (and similarly in the even sector).
As a consequence, the $(\ell,m)$ piece of $R_{\a \b}^{(-)}$ depends only on the $(\ell,m)$ piece of $h_{\a \b}^{(-)}$ (see equation (\ref{exp1})):
\begin{equation} \label{pert-l}
 h_{\a \b}^{(\ell,m,-)} =  \left( \begin{array}{cc} 0 &  h_a^{(\ell,m)} \widehat \e _B{}^C \hd_C S_{(\ell,m)} \\   h_b^{(\ell,m)} \widehat \e _A{}^C  
\hd_C S_{(\ell,m)}& 2 r^2 k^{(\ell,m)} \widehat \e_{(A}{}^C  \hd_{B)}  \hd_C  S_{(\ell,m)}
 \end{array} \right),
\end{equation}
and similarly in the even sector.  
Note from (\ref{pert-l}) and (\ref{kvfs}) that  the  odd $\ell=1$ modes add up to 
\begin{equation} \label{pl1-}
h_{\a \b}^{(\ell=1,-)} = \left( \begin{array}{cc} 0 & \sum_{m=1}^3 \sqrt{\tfrac{3}{4\pi}}\, h_a^{(\ell=1,m)} J_{(m) B} \\ 
\sum_{m=1}^3  \sqrt{\tfrac{3}{4\pi}}\, h_b^{(\ell=1,m)} J_{(m) A} & 0 \end{array} \right). 
\end{equation}
These $\ell=1$ perturbations correspond to  infinitesimal rotation, i.e., to deformations towards a stationary Kerr solution. \\

The $(\ell,m,+)$ mode is defined in a way analogous to  (\ref{pert-l}), i.e.,  keeping a single term in the spherical harmonic expansion 
of the even fields $(h_{ab}, q_a, J, G)$. It is important to note that 
\begin{align}
\mathbf{J}^2  h_{\a \b}^{(\ell,m,\pm)} &= - \ell (\ell+1) h_{\a \b}^{(\ell,m,\pm)}, \\
 P_* h_{\a \b}^{(\ell,m,\pm)} &= \pm (-1)^{\ell} 
h_{\a \b}^{(\ell,m,\pm)}.
\end{align}
The different behavior under parity is the signature that distinguishes odd from even modes.

\subsection{The Schwarzschild (A)dS solution}

The Schwarzschild / Schwarzschild (anti) de-Sitter (S(A)dS) is the only 
spherically symmetric  solution  of the vacuum Einstein equation with a 
cosmological constant $\Lambda$:
\begin{equation}
R_{\a \b} = \Lambda g_{\a \b},
\end{equation}
$R_{\a \b}$  the Ricci tensor of the Lorentzian metric $g_{\a \b}$.
The manifold is $\mathbb{R}_v \times (0,\infty)_r \times S^2$,  the metric (c.f. equation (\ref{ssm}))
\begin{equation} \label{sads}
ds^2 = - f dv^2 + 2 dv \; dr + r^2 \left( d\theta^2 + \sin^2 \theta \; d\phi^2 \right), \;\;\,  f = 1-\frac{2M}{r}-\frac{\Lambda r^2}{3}.
\end{equation} 
The constant 
 $M$ is the mass of the solution,  $M=0$ in (\ref{sads}) gives the  Minkowski ($\Lambda=0$), de Sitter ($\Lambda >0$) and anti de Sitter 
($\Lambda <0$) spacetimes.  $M>0$  in (\ref{sads}) corresponds to    a Schwarzschild black hole if  $\Lambda=0$, a
Schwarzschild de Sitter (SdS) black hole if $\Lambda>0$ and   $\Lambda < (3M)^{-2}$, 
and SAdS black hole if $\Lambda <0$. \\

The Killing vector field $k^a \p_a  = \p/\p_v$ 
is timelike in the open sets  defined by $f>0$ and spacelike in the open sets defined by $f<0$. The null hypersurfaces of constant 
{\em positive} $r$ where $f=0$ are the horizons, they cover the curvature singularity at $r=0$. \\

In any open set where $f \neq 0$  we may define a ``tortoise'' radial coordinate $r^*$ by
\begin{equation} \label{rs}
\frac{dr^*}{dr} = \frac{1}{f},
\end{equation}
and a coordinate $t$ through 
\begin{equation}
t= v-r^*. 
\end{equation}
The metric in 
 static coordinates $(t,r,\theta,\phi)$ is 
\begin{equation}
ds^2 = -f dt^2 + \frac{dr^2}{f} + r^2 \left( d\theta^2 + \sin^2 \theta \; d\phi^2 \right). 
\end{equation}

\subsubsection{Horizons and the static region}
The Schwarzschild  ($\Lambda=0$) and the SAdS ($\Lambda<0$) black holes have  
 a single horizon at  $r=r_h$, satisfying 
\begin{equation}
 r_h  = 2M +  \tfrac{1}{3} \Lambda {r_h}^3 \leq 2M,
\end{equation}
in terms of which
\begin{equation}
f = 1 - \frac{(1-\tfrac{1}{3} \Lambda {r_h}^2) r_h}{r}-\frac{\Lambda r^2}{3}.
\end{equation}
These black holes  have a non-static region I defined by $0<r<r_h$ and a static region II defined by $r_h < r < \infty$.\\

The SdS black holes are those for which $0 < \Lambda < (3M)^{-2}$ ($f$ has  single real root when $\Lambda > (3M)^{-2}$, and this root  is negative). 
For SdS black holes   $f$ has 
one negative ($r_o$) and two positive roots $r_h <r_c$, $r_o+r_h+r_c=0$, with $r_c \to r_h{}^+$ as $9M^2 \Lambda \to 1^-$, and \cite{Lake}
\begin{equation} \label{rhrc}
2M < r_h < 3M < r_c.
\end{equation}
There is a non-static region I defined by $0<r<r_h$ adjacent to a static region II ($r_h<r<r_c$),  and a further 
non-static region III defined by $r>r_c$. \\

This paper focuses on the stability of the static region II of Schwarzschild and S(A)dS  black holes. 

\subsubsection{The bifurcation sphere at $r=r_h$} \label{bss}

Let $g$ be  the function 
\begin{equation}
g(r)=\frac{1}{f(r)} - \frac{1}{f'(r_h) (r-r_h)}.
\end{equation}
 $g$  is smooth in I $\cup$ II (i.e., $0<r<r_c$ for SdS, $0<r$ for Schwarzschild and SAdS), since $f$ is smooth with  a simple zero at $r=r_h$
 and no zeros in I and II. \\
Consider the following solution of (\ref{rs}) 
\begin{equation} \label{srs}
r^* = G(r)+ \frac{1}{f'_h} \ln \left| \frac{r}{r_h}-1 \right|, \;\; \frac{dG}{dr}=g, \;\; f'_h = f'(r_h)
\end{equation}
The most general solution is obtained by adding (possibly different) constants to the left and right of $r=r_h$, however, for 
$r^*$ as in (\ref{srs}), the function 
\begin{equation}
s(r) = \left( r/r_h -1 \right)  e^{f'_h G} = \begin{cases} e^{f'_h r^*} & \text{ , in II } \\  -e^{f'_h r^*} & \text{ , in I }  \end{cases}
\end{equation}
is smooth in I $\cup$ II, and  monotonically growing, so it has an inverse $r=\mathcal{K}(s)$. Introduce $u=t-r^*$ in addition 
to $v=t+r^*$ above, then 
\begin{equation}
ds^2 = -f(r) du dv + r^2 \left( d\theta^2 + \sin^2 \theta \; d\phi^2 \right). 
\end{equation}
Now let 
\begin{equation}
(U,V) = \begin{cases} (-e^{-f'_h u/2}, e^{f'_h v/2}) & \text{ , in II } \\  (e^{-f'_h u/2}, e^{f'_h v/2}) & \text{ , in I } \end{cases} 
\end{equation}
Note that $UV=-s(r)$ and $dU dV = \tfrac{1}{4} {f'_h}^2  s(r) \; du dv$, therefore (\ref{sads}) is equivalent to 
\begin{equation} \label{UVmetric}
ds^2 = \frac{4 f(r)}{{f'_h}^2 UV}\; dU dV + r^2  \left( d\theta^2 + \sin^2 \theta \; d\phi^2 \right), \;\; r = k(-UV),\;  \;\;k(-UV)>0
\end{equation}
and $V>0$. 
We may now take two extra copies of I and II (call these $\rm{I}'$ and $\rm{II}'$) and define 
\begin{equation}
(U,V) = \begin{cases} (e^{-f'_h u/2}, -e^{f'_h v/2}) & \text{ , in } \rm{II}' \\  (-e^{-f'_h u/2}, -e^{f'_h v/2}) & \text{ , in  } \rm{I}' \end{cases} 
\end{equation}
The black hole metric  (\ref{sads})  in $\rm{I}' \cup \rm{II}'$ is again given by equation (\ref{UVmetric}), 
except that now   $V<0$. 
It can be checked that (\ref{UVmetric}) is  smooth on $S^2$ times the open region of the $(U,V)$ plane defined 
by $\mathcal{K}(-UV)>0$   ($0<\mathcal{K}(-UV)<r_c$ for SdS). This region contains two copies of I and two copies of II. 
Since $r$ is a function of the product $UV$, the metric (\ref{UVmetric}) has the discrete $Z_2$ symmetry 
$(U,V) \to (-U,-V)$ under which 
${\rm I} \leftrightarrow {\rm I}'$ and  ${\rm II} \leftrightarrow {\rm II}'$. 
The $Z_2$ invariant set $U=V=0$   is a sphere of radius $r_h$, called  {\em bifurcation sphere}. \\

Integrating (\ref{rs}) in region II we find that, 
for $\Lambda=0$, after choosing an integration constant, 
\begin{equation} \label{rs0}
r^* = r + 2M \;\ln\left( \frac{r}{2M}-1 \right);
\end{equation}
for $\Lambda<0$  the integration constant can be chosen such that 
\begin{equation} \label{srnl}
r^* = - \int_r^{\infty} \frac{dr'}{f(r')} \simeq \begin{cases} \frac{r_h}{1-\Lambda {r_h}^2} \ln \left( \frac{r}{r_h}-1 \right) &, r \to r_h^+ \\
\frac{3}{\Lambda r} &, r \to \infty 
\end{cases} , 
\end{equation}
and for $0< 9M^2 \Lambda<1$ and $r_h<r<r_c$, 
\begin{equation} \label{rs+L}
r^* \simeq \begin{cases} \frac{r_h (r_h r_c +{r_h}^2+{r_c}^2)}{(r_c-r_h)(2r_h+r_c)} \ln \left( \frac{r}{r_h}-1 \right) &, r \to r_h^+ \\
\frac{r_c (r_h r_c +{r_h}^2+{r_c}^2)}{(r_h-r_c)(2r_c+r_h)} \ln \left( \frac{r}{r_c}-1 \right)  &, r \to r_c^- 
\end{cases}.
\end{equation}
It follows that, in region II of a Schwarzschild or SdS black hole, $-\infty < r^* < \infty$ (which corresponds to the 
entire quadrants $U<0$, $V>0$ and $U>0$, $V<0$). For SAdS, on the other hand, 
$-\infty < r^* < 0$,  therefore 
$U >-1/V$ in  II, and $U<-1/V$ in $\rm{II}'$.\\ 

The above construction gives the Penrose-Carter diagrams in Figure 1. In the AdS case (right side of the figure), there is a conformal time like 
boundary corresponding to 
$U >-1/V$ in  II, and $U<-1/V$ in $\rm{II}'$. The conformal boundary of region II 
is replaced in the Schwarzschild case by null infinity, and in the SdS case  by the cosmological horizon 
(left side of the figure). 
For the SdS black hole it is possible to 
follow a procedure similar to the one outlined above that gives a further  extension  with a bifurcation sphere at $r=r_c$,
 two copies of region  III, and an extra copy of II. 
The maximal analytic extension of SdS contains infinitely many copies of I, II, and III \cite{Gibbons:1977mu}.

\begin{figure}
\begin{center}
\includegraphics[width=15cm]{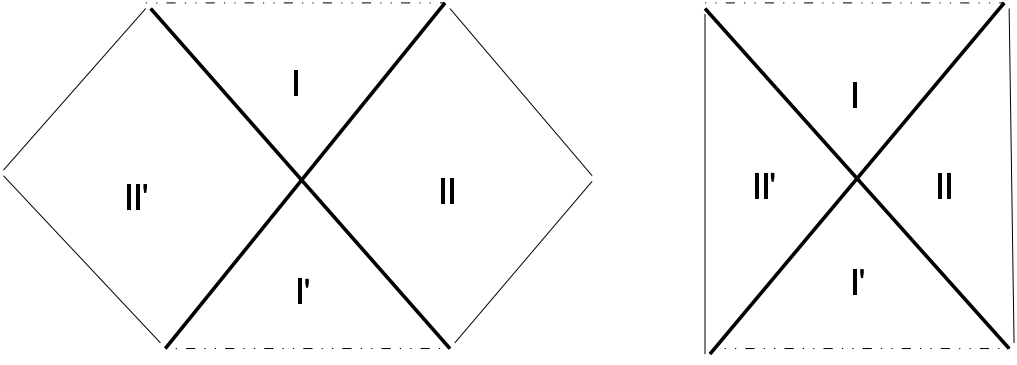}
\caption[Bifu]{Left: Penrose-Carter diagram for a Schwarzschild black hole (SdS black hole). The dotted lines are the $r=0$ 
singularities, the thick 
lines are  past  and future event horizons -the $U$ and $V$ axis in (\ref{UVmetric})-, they intersect at the bifurcation sphere. The thin lines correspond to null infinity (the cosmological horizon). \\
Right: Penrose-Carter diagram for a SAdS black hole. The dotted lines are the singularities at $r=0$, the thick 
lines are the past and future event horizons -the $U$ and $V$ axis intersecting  at the bifurcation sphere-. The thin  time-like lines correspond to  null infinity. } \label{figurita}
\end{center}
\end{figure}

\subsubsection{Hyperbolic equations and global hyperbolicity of the static region} \label{hyp}

If we use the coordinates  $(t,r^*, \theta,\phi)$, the 
  scalar wave equation (\ref{lap}) acquires  the simple  form
\begin{equation} \label{lap2.1}
\nabla^{\a} \nabla_{\a} \Phi = \frac{1}{rf}
  \left[ -\p_t^2 + \p_{r^*}^2 + f \left( \frac{\hd^A \hd_A}{r^2} - \frac{2M}{r^3} + \frac{2 \Lambda}{3}
\right) \right] (r\Phi),
\end{equation}
and  (\ref{4DRWE}) reduces to
\begin{equation} \label{lap2}
0 = -rf\left[ \nabla^{\a} \nabla_{\a} \Phi + \left( \frac{8M}{r^3} - \frac{2 \Lambda}{3}  \right) \Phi \right] 
  =   \left[ \p_t^2 - \p_{r^*}^2 - f \left( \frac{\hd^A \hd_A}{r^2} + \frac{6M}{r^3}
\right) \right] (r\Phi),
\end{equation}
that is,
\begin{equation} \label{2+2}
\frac{\p^2}{\p t^2} \left( r \Phi \right) + \mathcal{A}^{\Lambda} \left( r \Phi \right) =0,
\end{equation}
where 
\begin{equation} \label{AL}
\mathcal{A}^{\Lambda} = -\frac{\p^2}{\p {r^*}^2} + \left( 1 - \frac{2M}{r} - \frac{\Lambda}{3} r^2 \right) \left( -\frac{6M}{r^3} - \frac{\hd^A \hd_A}{r^2} \right) 
=: -\frac{\p^2}{\p {r^*}^2} + V_1 - V_2 \hd^A \hd_A, 
\end{equation}
and
\begin{equation} \label{V12a}
V_1 = \left( 1-\frac{2M}{r}- \frac{\Lambda}{3} r^2 \right) \left( -\frac{6M}{r^3}\right)\;\; \text{ and  }\;\;
V_2 = \left( 1-\frac{2M}{r} - \frac{\Lambda}{3} r^2\right) \left( \frac{1}{r^2} \right).   
\end{equation}
If we expand $\Phi$ in spherical harmonics, 
\begin{equation} \label{4DRW2D}
 \Phi = \sum_{(\ell \geq 0,m)} \frac{\phi_{(\ell,m)}}{r} S_{(\ell,m)} : \mathcal{M} \to \mathbb{R}
\end{equation}
and use this in equations (\ref{2+2})-(\ref{V12a}), we find that (\ref{4DRWE}) is indeed equivalent to the set of equations  (\ref{RWE}), 
extended to include $\ell=0,1$.\\

Note that, although the above equations are formally similar for different values of $\Lambda$,
  $-\infty < r^* < \infty$ as $r$ spans region II of the Schwarzschild and SdS black holes, but is restricted to 
$r^*<0$  in region II  of the SAdS black hole.  
The reason why equation (\ref{4DRWE}) for SAdS 
is equivalent to a system of wave equations with a potential in {\em a half} of a two dimensional 
Minkowski space,   is that 
region II of SAdS is not globally hyperbolic, but has a conformal   timelike boundary at $r^*=0$ (equation (\ref{srnl})).
The problem of defining the dynamics imposed by hyperbolic equations (such as 
 (\ref{4DRWE}), (\ref{RWE}) and (\ref{ZE})) in a non globally hyperbolic static spacetime (either 
SAdS or its two dimensional orbit space SAdS$/SO(3)$), that is, giving 
 unique solutions 
from initial data at a surface $\Sigma_o$ transverse to the timelike Killing vector field $k^a$, has been  addressed in 
the series of articles \cite{Wald:1980jn}, \cite{Ishibashi:2003jd} and \cite{Ishibashi:2004wx}. As expected, the dynamics is unique only 
within the domain of dependence of the support of the initial datum; outside it, it depends on the self 
adjoint extension that we choose for the operator $\mathcal{A}$  which is roughly our choice of 
boundary conditions at  $r^*=0$. \\

For the Zerilli equation (\ref{ZE}) of an SAdS black hole, we have found that there is a choice of boundary conditions at $r^*=0$  
under which even perturbations that preserve the AdS asymptote, grow exponentially with time. For this choice,
 SAdS is unstable. There is an infinite  set of possible  boundary conditions 
for equation (\ref{ZE}) in SAdS$/SO(3)$, we are currently analyzing the dynamics and stability 
of SAdS black holes under different choices \cite{bernardo}.

\section{The Linearized Einstein equation: odd sector}

There are many sources for the well know solution of the LEE around a Schwarzschild background. Two historically relevant 
references are \cite{Moncrief:1974am}, where the gauge invariance of RW and Zerilli potentials was established, and \cite{Gerlach:1979rw}, 
where a ``covariant'' formulation using arbitrary coordinates {\em for the orbit space} in order to study, e.g., 
waves crossing the horizon  was developed (this should more properly be called ``2D covariant approach''). 
More recently, the 2D covariant approach was generalized to GR black holes with constant curvature horizons in arbitrary dimensions 
and including a cosmological constant in a series of papers by Ishibashi and Kodama (see, e.g., \cite{Ishibashi:2003ap} 
and \cite{Kodama:2003jz}). 
 We find 
the 2D covariant approach in reference 
\cite{Sarbach:2001qq}  particularly well adjusted to 
our purposes  as, 
although  restricted to the dynamical $\ell \geq 2$ perturbations, avoids unnecessary expansions in harmonic modes. Section V in \cite{Sarbach:2001qq} 
is devoted to the Schwarzschild black hole, we generalize the exposition there to $\Lambda \neq 0$ below.

\subsection{Gauge transformations and gauge invariants}
The effect of the diffeomorphism generated by the vector field $\xi^{\a}$ on 
the first order variation of a tensor $T$  is $\dot T \to \dot T + \pounds_{\xi} T$. In particular, 
the effect on $h_{\a \b}= \dot g_{\a \b}$, 
\begin{equation} \label{gt}
h_{\a \b} \to \pounds_{\xi} g_{\a \b} = h_{\a \b} + \nabla_{\a} \xi_{\b} + \nabla_{\b} \xi_{\a} = h'_{\a \b},
\end{equation}
defines an equivalence relation where $h_{\a \b} \sim h'_{\a \b}$ iff there exists a covector field $\xi_{\a}$ such that 
(\ref{gt}) holds. 
To measure the strength  of the perturbation we will analyze gauge invariant fields, i.e., fields that
 depend only  on the class $[h_{\a \b}]$ of a metric perturbation under the equivalence 
relation defined above. In this section we show how to parametrize the set of classes  $[h_{\a \b}]$ by means of gauge invariant fields and how to  
 choose a class representative (gauge fixing). \\

Given a generic vector field (\ref{vd})-(\ref{0l0}) we find that 
\renewcommand*{\arraystretch}{1.8}
\begin{equation} \label{gto}
\nabla_{\a} \xi_{\b}^{(-)} + \nabla_{\b} \xi_{\a}^{(-)} = 
\left( \begin{array}{cc} 0 & \widehat \e _B{}^C \hd_C (r^2 \td_a Y ) \\ \widehat \e _A{}^C \hd_C (r^2 \td_b Y ) & 2 r^2 \widehat \e_{(B}{}^C \hd_{A)} \hd_C Y
\end{array} \right)
\end{equation}
and 
\renewcommand*{\arraystretch}{1.8}
\begin{equation} \label{gte}
\nabla_{\a} \xi_{\b}^{(+)} + \nabla_{\b} \xi_{\a}^{(+)} = 
\left( \begin{array}{cc} \td_a \xi_b+ \td_b \xi_a & \hd_B (r^2 \td_a X + \xi_a) \\ \hd_A (r^2 \td_b X + \xi_b)  & 2 r^2 \hd_A \hd_B X + 2 r (\td_a r) \xi^a 
\widehat g_{AB}
\end{array} \right)
\end{equation}
where it should be kept in mind that $X$ and $Y$ belong to $L^2(S^2)_{>0}$.
These equations imply that the odd (even) piece of $\xi_{\a}$ affects the odd (even) piece of $h_{\a \b}$ in (\ref{gt}). 
Let us analyze  the effect of gauge transformations in the odd sector. Inserting $Y = Y_{(\ell=1)} + Y^{>1}$,
$$ Y_{(\ell=1)} = \sum_{m=1}^3 Y_{(m)}(x) S_{(\ell=1,m)}(\theta,\phi), \;\;\; Y^{>1} = \sum_{\ell\geq 2, m}^3 Y_{(\ell,m)}(x) S_{(\ell=1,m)}(\theta,\phi),$$
in (\ref{gto}) ($x$ are coordinates of the orbit manifolds, $S_{(\ell,m)}$ orthonormal real spherical harmonics), 
defining $h_{\a \b}^{(-,>1)} = h_{\a \b}^{(-)}- h_{\a \b}^{(\ell=1,-)}$ and splitting similarly   all the other fields, we find from (\ref{pl1-}), (\ref{gt}) and (\ref{kvfs}) that 
\begin{equation} \label{gt-1}
{h'}_{\a \b}^{(\ell=1,-)} = \left( \begin{array}{cc} 0 & \sum_{m=1}^3 \sqrt{\tfrac{3}{4\pi}}\, \left(h_a^{(\ell=1,m)} + r^2 \td_a Y_{(m)} \right) J_{(m) B} \\
 \sum_{m=1}^3 \sqrt{\tfrac{3}{4\pi}}\,  \left(h_b^{(\ell=1,m)}+ r^2 \td_b Y_{(m)} \right) J_{(m) A} & 0 \end{array} \right)
\end{equation}
and also that 
\begin{equation} \label{pert-gt}
 {h'}_{\a \b}^{(-,>1)} =  \left( \begin{array}{cc} 0 &  \widehat \e _B{}^C \hd_C \left(h_a^{>1} + r^2 \td_a Y^{>1} \right) \\  \widehat \e _A{}^C  \hd_C 
\left(h_b^{>1} + r^2 \td_b Y^{>1} \right)  & 2 r^2\widehat \e_{(A}{}^C  \hd_{B)}  \hd_C (k^{>1} + Y^{>1})
 \end{array} \right)
\end{equation}
This last equation implies that: i) $H_a = h_a^{>1}- r^2 \td_a k^{>1}$ is a gauge invariant field and ii) there is a gauge (commonly known as 
the Regge-Wheeler gauge) where the $\ell >1$ piece of the odd metric perturbation assumes the form 
\begin{equation} \label{pert-RW}
 h_{\a \b}^{(-,>1)} =  \left( \begin{array}{cc} 0 &  \widehat \e _B{}^C \hd_C H_a \\  \widehat \e _A{}^C  \hd_C H_b & 0 \end{array} \right),
\;\;\; H_a \in L^2(S^2)_{>1}
\end{equation}

\subsection{Solution of the linearized Einstein equation}

A calculation of the LEE (\ref{lee}) using (\ref{pert-RW}) gives, in agreement with  equations (31) in \cite{Chaverra:2012bh}, 
\begin{align} \label{dr1}
\dot R_{ab} - \Lambda h_{a b} &= 0\\
\dot R_{aB} - \Lambda h_{a B} &= -\frac{1}{2 r^2} \widehat \e_B{}^C \hd_C \left[ \tilde \e_a{}^c \td_c (r^2 \mathcal{F}) - H_a \, \td^c \td_c(r^2)  +\hd^C 
\hd_C H_a + 2 \Lambda r^2 H_a \right]\\
\dot R_{AB} - \Lambda h_{A B} &= \widehat \e_{(A}{}^C \hd_{B)} \hd_C (\td^a H_a) 
\end{align}
where 
\begin{equation} \label{f}
\mathcal{F} = r^2  \tilde \e^{ab} \td_a \left(\frac{H_b}{r^2} \right).
\end{equation}
In deriving equations (\ref{dr1})-(\ref{f}) we have not  made any assumptions on $H_a$ in (\ref{pert-RW}), so these calculations
 also hold for the $\ell=1$ mode given in 
(\ref{pl1-}) and (\ref{kvfs}) if we replace $H_a$ with $\sum_m h_a^{(\ell=1,m)}(x) S_{(1,m)}$. 
In this case, however, 
$$\left( - \td^c \td_c(r^2)  +\hd^C \hd_C  + 2 \Lambda r^2 \right) h_a^{(\ell=1,m)}S_{(1,m)}=0,$$
and  also   $\widehat \e_{(A}{}^C \hd_{B)} \hd_C ( S_{(1,m)}) =0$,  so the LEE for the odd $\ell=1$ modes can entirely 
be written in terms of the  field 
\begin{equation} \label{f1}
\mathcal{F}_1 = r^2  \tilde \e^{ab} \td_a \left(\sum_m \frac{h_b^{(\ell=1,m)}}{r^2}  \right) S_{(1,m)}
\end{equation}
which is invariant under the gauge transformation (\ref{gt-1}). Thus, for the odd $\ell=1$ mode we find 
\begin{align} \label{drl1}
\dot R_{ab} - \Lambda h_{a b} &= 0\\
\dot R_{aB} - \Lambda h_{a B} &= -\frac{1}{2 r^2} \widehat \e_B{}^C \hd_C \left[ \tilde \e_a{}^c \td_c (r^2 \mathcal{F}_1)  \right]\\
\dot R_{AB} - \Lambda h_{A B} &= 0.
\end{align}
Using the facts that the kernel of the operator $\widehat \e_{A}{}^C  \hd_C$ is the $\ell=0$ mode and the kernel of 
$\widehat \e_{(A}{}^C \hd_{B)} \hd_C$ are the $\ell=0,1$ modes, 
 we arrive at

\begin{lem}
The  LEE (\ref{lee}) in the odd sector is equivalent  to the system
\begin{align} \label{1}
&\td^a H_a =0\\ \label{two}
& \tilde \e_b{}^c \td_c (r^2 \mathcal{F}) - H_b \, \td^c \td_c(r^2)  +\hd^C 
\hd_C H_b + 2 \Lambda r^2 H_b =0 \\
&  \td_c (r^2 \mathcal{F}_1)=0 \label{l1-e}.
\end{align}
\end{lem}

Let us treat first the equations for $\ell=1$ modes. The general solution of equation (\ref{l1-e}) is \cite{Sarbach:2001qq} 
\begin{equation} \label{sol-1}
h_a^{(\ell=1,m)} = \frac{ 2 j^{(m)}}{r N(r)} \tilde \e_a{}^b \td_b r + r^2 \td_a \left( \frac{Z^{(m)}}{r^2} \right), \;\; N(r) = g^{ab} (\td_a r) (\td_b r),  
\end{equation}
$j^{(m)}$ a gauge invariant constant and  $Z^{(m)}$ an arbitrary function that we recognize from (\ref{gt-1}) as ``pure gauge''. For the Schwarzschild-(A)dS 
space-time in Schwarzschild coordinates, $g_{ab} dx^a dx^b = -f(r) dt^2+ dr^2/f(r)$, $ f(r)=1-2M/r-\Lambda r^2/3= N(r)$, and equation 
(\ref{sol-1})  reads 
\begin{equation} \label{sol-1b}
h_a^{(\ell=1,m)} dx^a = -\frac{2 j^{(m)}}{r} dt + r^2 \td_a \left( \frac{Z^{(m)}}{r^2} \right) \; dx^a.
\end{equation}
If we take, e.g., $j^{(1)}=j^{(2)}=0, j^{(3)}=aM$, choose the gauge 
$Z^{(1)}=Z^{(2)}=0$, $Z^{(3)}=-\Lambda (a /3)\, r^2 t$,  and insert the above equation in (\ref{pl1-}), we obtain
\begin{equation} \label{pl1-b}
h_{\a \b}^{(\ell=1,-)}  dx^{\a} dx^{\b}  = -2  \left(\frac{2aM}{r} +\frac{\Lambda a}{3} r^2 \right) \; \sin^2 (\theta)  \; dt \,  d\phi, 
\end{equation}
which we recognize as the first order  Taylor expansion around $a=0$ of the Kerr-(A)dS metric with angular momentum $J=aM$ in Boyer-Lindquist 
coordinates. This proves our previous assertion 
that $\ell=1$ odd modes correspond to displacements within the stationary Kerr family of black holes. \\

Consider now the $\ell >1$ odd LEE, equations (\ref{1}) and (\ref{two}). For the  Schwarzschild-(A)dS background equation (\ref{two}) reads
\begin{equation} \label{2p}
\mathcal{E}_b = \tilde \e_b{}^c \td_c(r^2 \mathcal{F}) + (2+ \hd^A \hd_A) H_b =0,
\end{equation}
and, as noticed in \cite{Chaverra:2012bh},   implies (\ref{1}) (proof: $0= \td^b \mathcal{E}_b = (2+ \hd^A \hd_A) \td^b H_b$  implies 
$\td^b H_b=0$ since $H_b \in L^2(S^2)_{>1}$). On the other hand, equation (\ref{2p}) implies ({\em but is not equivalent to!})
$\tilde \e^{ab} \td_a (r^{-2} \mathcal{E}_b) =0$, a condition that, 
from the definition of $\mathcal{F}$, equations (\ref{lap}) and   (\ref{f}) and the background equation
\begin{equation} \label{rab}
 \td_a \td_b r =  \left( \frac{M}{r^2} - \frac{\Lambda}{3} r \right)  \; \tilde g _{ab},
\end{equation}
is seen to be equivalent to 
\begin{equation} \label{rwf}
\nabla^{\a} \nabla_{\a} \mathcal{F} + \left(  \frac{8M}{r^3} - \frac{2 \Lambda}{3} \right) \mathcal{F} =0.
\end{equation}
Interestingly 
enough,  (\ref{l1-e}) and (\ref{f1})  imply that  $\mathcal{F}_1= \sum_{m=1}^3 C_{(m)} S_{(1,m)}(\theta,\phi)/r^2$
 which, using (\ref{lap}) is easily seen to also satisfy equation (\ref{rwf}). We conclude that the LEE implies that 
\begin{equation} \label{rweF}
\nabla^{\a} \nabla_{\a} F + \left(  \frac{8M}{r^3} - \frac{2 \Lambda}{3} \right) F =0, \;\;\; (F = \mathcal{F}_1+\mathcal{F}).
\end{equation}
Back to  (\ref{1})-(\ref{two}), we note that the solution of (\ref{1}) is
\begin{equation} \label{Ha}
H_a = \tilde \e_{a b} \td^b (r^2\Phi), 
\end{equation}
for some scalar field $\Phi \in L^2(S^2)_{>1}$. This implies that 
\begin{equation}
\mathcal{F} = r^2 \td^c \left( r^{-2} \td_c(r^2 \Phi) \right).
\end{equation}
Inserting (\ref{Ha})  in (\ref{2p}) gives 
\begin{equation} \label{2pp}
\tilde \e_b{}^c \td_c \left[ r^2 \mathcal{F} +  (2+ \hd^A \hd_A) r^2 \Phi \right] =0.
\end{equation}
The general solution to this equation is 
\begin{equation}
\frac{\mathcal{F}}{r^2} +  \frac{2+ \hd^A \hd_A}{ r^2} \Phi = \td^c \left( r^{-2} \td_c(r^2 \Phi) \right) +  \frac{2+ \hd^A \hd_A}{ r^2} \Phi= 
\frac{z(\theta,\phi)}{r^4},
\end{equation}
and we may (and will) choose $z(\theta,\phi)=0$, which gives 
\begin{equation} \label{FP}
\mathcal{F} = - [2+ \hd^A \hd_A] \Phi   = - [2+ \mathbf{J}^2] \Phi,
\end{equation}
together with 
\begin{align} \label{2drwe}
0 &= \td^c \left( r^{-2} \td_c(r^2 \Phi) \right) +  \frac{2+ \hd^A \hd_A}{ r^2} \Phi \\ \nonumber
   &= \td^c \td_c \Phi + \frac{\hd^A \hd_A}{ r^2} \Phi  +\frac{2}{r} \td^c r \td_c \Phi +  \frac{2}{r^2} (1-\td^c r \td_c r + r \td^c \td_c r) \Phi,  
\end{align}
which, using again (\ref{lap}) and (\ref{rab}), is seen to be equivalent to  the 4DRW equation 
\begin{equation} \label{rweP}
\nabla^{\a} \nabla_{\a} \Phi + \left(  \frac{8M}{r^3} - \frac{2 \Lambda}{3} \right) \Phi =0, \;\;\; \Phi \in L^2(S^2)_{>1}.
\end{equation}
Equation (\ref{rweP}) is equivalent to the standard form (\ref{RWE}) after decomposing $\Phi$ in modes as in (\ref{4DRf}). \\
We can gather our results concerning the odd sector of the LEE in the following:
\begin{lem} \label{oddsol}
Consider the odd sector of the LEE around a Schwarzschild-(A)dS background:
\begin{itemize}
\item[(i)] 
The solution of the $\ell>1$ piece of the metric perturbation 
in the Regge-Wheeler gauge is given by equation (\ref{pert-RW}), where the gauge invariant field 
$H_a$ is as in equation (\ref{Ha}) and $\Phi$ satisfies the 4DRW equation (\ref{rweP}). Equations   (\ref{pert-RW}) and (\ref{Ha}) are equivalent to (\ref{CsP}). 
\item[(ii)] The solution of the $\ell=1$ piece  is given in equations (\ref{pl1-}) and (\ref{sol-1b}). In this last equation 
$j^{(m)}$ are three gauge invariant constants and $Z^{(m)}$ three arbitrary gauge functions.
\item[(iii)] There is a bijection between the  space $\mathcal{L}_-$ of smooth odd solutions of the LEE mod gauge transformations and 
the set $$L_- = \{ j^{(m)}, m=1,2,3 \} \cup \{ \Phi \; | \;  \Phi \text{ is a smooth solution of equation (\ref{rweP})} \}.$$ 
\item[(iv)] The LEE implies that the gauge invariant field $F = \mathcal{F}_1 + \mathcal{F}$ (defined in equations (\ref{f}) and (\ref{f1})) 
satisfies the 4DRW equation (\ref{rweF}).
\end{itemize}
\end{lem}
The equivalence of  (\ref{pert-RW}) and (\ref{Ha})  to (\ref{CsP}) is checked by  a straightforward calculation.
 (iii)  follows immediately from (i) and (ii). Note also  that the relation (\ref{FP}) explains why the fields $\mathcal{F}$ and $\Phi$, 
which belong to  $L^2(S^2)_{>1}$,  satisfy the same wave 
equation: the operators $\nabla^{\a} \nabla_{\a}  + \left(  \frac{8M}{r^3} - \frac{2 \Lambda}{3} \right)$ and $[2+ \hd^A \hd_A] $ commute, 
and $[2+ \hd^A \hd_A]$  is invertible in $L^2(S^2)_{>1}$ (on scalar fields, $\hd^A \hd_A = \mathbf{J}^2$ and on arbitrary tensor fields
 $[\mathbf{J}^2,\nabla_{\a}]=0$,  
since $[\pounds_K, \nabla_{\a}]=0$ for any Killing vector field $K^{\a}$).\\

\subsection{Measurable effects of the perturbation on the geometry}

 For the purposes of a non-modal analysis, Lemma \ref{oddsol}.(iii) offers a  more appropriate parametrization of 
the dynamical sector of the odd perturbations  than the infinite set of $\phi_{(\ell,m)}^-$'s. 
 However,  unlike   the $\ell=1$ parameters $j_{(m)}$,  no clear geometrical meaning can be attached   to $\Phi$ in (\ref{rweP}), beyond 
that of being a potential for solutions of the LEE, equation (\ref{CsP}). 
Consider now the curvature scalars in  (\ref{curvaturescalars}). 
In this section we show that the first order variation  $G_-=\dot Q_-$ of $Q_-$ 
associated to a perturbation $[h_{\a \b}^{(-)}] \in \mathcal{ L}_-$ 
contains all the gauge invariant information about $h_{\a \b}^{(-)}$. In particular, $\Phi$ and the $j^{(m)}$can be recovered 
from $G_-=\dot Q_-$ which, unlike $\Phi$, has a distinct geometrical meaning. Note that, since $Q_-=0$ for 
the S(A)dS geometries, $\dot Q_-$ is gauge invariant. \\

$\dot Q_-$ can be obtained from the set of equations (3) and (29) in \cite{Chaverra:2012bh}. The calculations are 
tedious and not particularly illuminating. The result is 
\begin{equation} \label{qd}
\dot Q_- = \frac{3M}{r^5} \hd^A \hd_A F = \frac{3M}{r^5} \mathbf{J}^2 F, \;\;\;(F= \mathcal{F}_1 + \mathcal{F}).
\end{equation}
As an example, for the  choice $j^{(1)}=j^{(2)}=0, j^{(3)}=aM$ in (\ref{pl1-b}), the $\ell=1$ piece of $\dot Q$ in (\ref{qd}) is
\begin{equation} \label{q1a}
\dot Q_-^{(\ell=1)} = \frac{3M}{r^5}  \mathbf{J}^2  \mathcal{F}_1 = -\frac{6M^2 a}{r^7} \cos \theta, 
\end{equation}
which   agrees with the first order Taylor expansion of $Q_-$  for the Kerr-(A)de-Sitter black hole in Boyer-Lindquist coordinates  around $a=0$, as 
anticipated. 
For arbitrary $j^{(m)}$'s, equation (\ref{q1a}) generalizes to 
\begin{equation} \label{q1b}
\dot Q_-^{(\ell=1)}  = - \frac{6M}{r^7} \sqrt{\tfrac{4 \pi}{3}}  \sum_{m=1}^3 j^{(m)}S_{(1,m)}.
\end{equation}
Combining equations (\ref{FP}), (\ref{qd}) and (\ref{q1b}) gives
\begin{equation} \label{qdt}
\dot Q_- = - \frac{6M}{r^7} \sqrt{\tfrac{4 \pi}{3}}  \sum_{m=1}^3 j^{(m)}S_{(1,m)} - \frac{3M}{r^5} \mathbf{J}^2 (\mathbf{J}^2 +2) \Phi,
\end{equation}
which generalizes to $\Lambda \neq 0$ 
 equation (22) in \cite{Dotti:2013uxa}. 
\begin{thm} \label{biy-}
Let $[h_{\a \b}^{(-)}] \in \mathcal{ L}_-$ and  $\dot Q_-\left([h_{\a \b}^{(-)}]\right)$ be the  first order variation of  
$Q_-$ for the perturbation $[h_{\a \b}^{(-)}]$: 
\begin{itemize}
\item[(i)] The field $r^5 \dot Q_-$ is in  $ L^2(S^2)_{>0}$ and satisfies the 4DRW equation 
\begin{equation} \label{rweqd}
\left[ \nabla^{\a} \nabla_{\a}  + \left(  \frac{8M}{r^3} - \frac{2 \Lambda}{3} \right)\right] (r^5 \dot Q_-) =0
\end{equation}
\item[(ii)] The map $[h_{\a \b}^{(-)}] \to \dot Q_-\left([h_{\a \b}^{(-)}]\right)$ is invertible: it is possible to construct 
a  representative of $[h_{\a \b}^{(-)}]$ from $\dot Q_-\left([h_{\a \b}^{(-)}]\right)$. 
\end{itemize}
\end{thm}
\begin{proof}
\mbox{}
\begin{itemize}
\item[(i)] 
Equation (\ref{qdt}) proves that $\dot Q_-  \in L^2(S^2)_{>0}$. This is a consequence  of the facts that  there are no odd $\ell=0$ modes, only odd modes 
contribute to  $\dot Q_-$, 
and 
$\mathbf{J}^2 \dot Q_-\left([h_{\a \b}^{(-)}]\right)= \dot Q_-\left( \mathbf{J}^2[h_{\a \b}^{(-)}]\right)$.
 Alternatively, by Birkhoff's theorem in a cosmological background \cite{Schleich:2009uj},
 the only  possible spherically symmetric perturbation of a Schwarzschild-(A)dS black hole amounts to a change of 
the black hole mass, and this does not affect the unperturbed value $Q_-=0$. Since $\mathbf{J}^2$ commutes with the wave operator 
in (\ref{rweF}), it follows from (\ref{qd}) and Lemma \ref{oddsol}.(iv) that $r^5 \dot Q_-$ satisfies (\ref{rweqd}). 
\item[(ii)] According to equation (\ref{qdt}), 
from the $\ell=1$ coefficients 
in the spherical harmonic expansion of $\dot Q_-$ we obtain  the constants $j^{(m)}$. This allows us to construct 
the $\ell=1$ piece of the $Z^{(m)}=0$ representative of  $[h_{\a \b}^{(-)}]$     in (\ref{sol-1}) and use it in 
(\ref{pl1-}).
  The orthogonal  projection of $\tfrac{r^5}{3M} \dot Q_-$ onto 
$L^2(S^2)_{>1}$ 
is related to $\Phi$ through the operator $\mathbf{J}^2 (\mathbf{J}^2 +2)$ (equation (\ref{qdt})). This operator  is invertible in $L^2(S^2)_{>1}$, 
so we can recover $\Phi$ from $\dot Q_-\left([h_{\a \b}^{(-)}]\right)$ and use it in (\ref{Ha}) or (\ref{CsP}) 
to construct the Regge-Wheeler representative 
(\ref{pert-RW}) of the $\ell>1$ piece of $[h_{\a \b}^{(-)}]$. \\
\end{itemize}
\end{proof}
Equation (\ref{qdt}) defines a bijection between the set $L_- \sim \mathcal{L}_-$defined in Lemma \ref{oddsol}.(iii) and the set 
\begin{equation} \label{ll-}
\hat{\mathcal{L}}_ - = \bigg\{ \dot Q_-\left([h_{\a \b}^{(-)}]\right) \; \bigg| \; 
h_{\a \b}^{(-)} \text{ is a solution of the LEE} \bigg\}
\end{equation}
This bijection implies that 
$\dot Q_-$ contains all the relevant (i.e., gauge invariant) information on the metric perturbation that gave rise to it. 
The invertible  relations 
\begin{equation}
\Phi \xrightarrow{-(2 + \mathbf{J}^2)\;} \mathcal{F} \xrightarrow{ \mathbf{J}^2}  r^5 \dot Q^{>1}_- 
\end{equation}
and 
\begin{equation}
 F \xrightarrow{ \mathbf{J}^2} r^5 \dot Q_-,
\end{equation}
explain why all these fields obey the same way equation. Note that the addition of an $\ell=1$ piece $\Phi^{(\ell=1)}$ to the potential 
 $\Phi$ satisfying a 4DRW equation 
would be irrelevant, as its contribution to $H_a$ would vanish: $\tilde \e_{a b} \td^b (r^2\Phi^{(\ell=1)})=0$.

\subsection{Nonmodal linear stability of the $\Lambda \geq 0$ black holes}

Having found that the   scalar  gauge invariant  field $G_-=\dot Q_-$ that measures  the distortion of the curvature
  encodes all the information on a given odd perturbation, it is natural
to define the strength of the perturbation as the magnitude of this field. 
A key additional feature of $\dot Q_-$ is the fact $r^5 \dot Q_-$ satisfies equation  (\ref{4DRWE}) 
(see  (\ref{rweqd}) in Theorem \ref{biy-}.(i)). This will be used to prove that the magnitude of $\dot Q_-$, 
and thus 
the strength of the perturbation, can be 
bounded on the entire outer static region. This fulfills the requirements of our proposed  notion of non-modal linear stability.
It is important to note that Theorem \ref{biy-}  applies to arbitrary 
smooth perturbations, whereas any boundedness or fall-off condition will  require a restriction 
to perturbations evolving from data that behave  properly as $r \to \infty$ ($r \to r_c$) in the asymptotically 
flat (de Sitter)  case.

\subsubsection{The asymptotically flat case}

In the $\Lambda=0$ case, a simple pointwise boundedness statement for $G_- = \dot Q_-$ can be made by 
noting  that the proof of boundedness of a Klein Gordon field in  \cite{Kay:1987ax} 
holds for the 4DRW equation:

\begin{thm} \label{odd-0}
 For any smooth solution of the odd LEE which has  compact support on Cauchy surfaces of the Kruskal extension
$ \rm{I \cup II \cup I' \cup II'}$  of 
the Schwarzschild space-time (Figure \ref{figurita}), there exists a constant $K_-$ such that $|G_- |~<~K_- \; r^{-6}$ for $r>2M$.
\end{thm}
\begin{proof}
From (\ref{qdt}), $|\dot Q_- | \leq  C/r^7 + | r^5 \dot Q_{>1}|/r^5$, where $C>0$ is a constant that depends on the $j^{(m)}$'s 
and the field $r^5 \dot Q_{>1}$   satisfies the 4DRW equation
(\ref{rweP}), so we only need  concentrate on this field. 
The similarities between the 4DRWE equation   and the massive Klein Gordon equation dealt with in 
 \cite{Kay:1987ax} allow us to proceed  by adapting  the proof of Theorem 1 in \cite{Kay:1987ax}. 
The $Z_2$ symmetry arguments in \cite{Kay:1987ax} showing  that this theorem    reduces to 
 the {\em intermediate case} (treated in the appendix in \cite{Kay:1987ax}) apply verbatim to equation (\ref{rweP}). 
The intermediate  case is defined by adding   the requirements that    $r^5 \dot Q$ and its $T$ derivative (in coordinates 
$(T,X,\theta,\phi)$, where $T=(U+V)/2$ and $X=(V-U)/2$)
vanish on the bifurcation sphere (refer to Section \ref{bss}).
There remains  to check that the proof in the appendix in \cite{Kay:1987ax} applies to 
the wave equation (\ref{rweP}). 
To this end, 
we use (\ref{lap2})-(\ref{V12a}) to cast (\ref{rweqd}) in the exterior Kruskal wedge as 
\begin{equation} \label{2+2.1}
\frac{\p^2}{\p t^2} \left( r^6 \dot Q_{>1} \right) + \mathcal{A} \left( r^6 \dot Q_{>1} \right) =0,
\end{equation}
where 
\begin{equation} \label{A}
\mathcal{A} = -\frac{\p^2}{\p {r^*}^2} + \left( 1 - \frac{2M}{r} \right) \left( -\frac{6M}{r^3} - \frac{\hd^A \hd_A}{r^2} \right) 
=: -\frac{\p^2}{\p {r^*}^2} + V_1 - V_2 \hd^A \hd_A, 
\end{equation}
with $r^*$ defined in (\ref{rs}), 
\begin{equation}\label{V12}
V_1 = \left( 1-\frac{2M}{r} \right) \left( -\frac{6M}{r^3}\right)\;\; \text{ and  }\;\;
V_2 = \left( 1-\frac{2M}{r} \right) \left( \frac{1}{r^2} \right).   
\end{equation}
For  the Klein Gordon equation  dealt with in \cite{Kay:1987ax}, the differential equation assumes this same form with 
\begin{equation}
V_1^{KG} = \left( 1-\frac{2M}{r} \right) \left( \frac{2M}{r^3} + m^2 \right), \;\;\;
V_2^{KG} =  V_2, 
\end{equation}
where $m^2$ is the square of the mass of the Klein Gordon field (which appears with a wrong sign in equation (1) in 
\cite{Kay:1987ax}, but with the  correct sign in the 
equations in the  appendix). However, the proof in the appendix in \cite{Kay:1987ax} does not make use of the explicit forms  
of $V_1^{KG}$ and $V_2^{KG}$, but only on the facts that these functions are bounded on the exterior wedge $r \geq 2M$, and that 
$\mathcal{A}$ is a positive definite self adjoint operator on $L^2(\mathbb{R} \times S^2, dr^*  \sin \theta \; d \theta \; d\phi)$. 
Since  $V_1$ and $V_2$ defined in (\ref{A})-(\ref{V12}) are  bounded for $r>2M$, and $\mathcal{A}$ is positive definite 
{\em on the $\ell>1$ subspace of $L^2(\mathbb{R} \times S^2, dx \; \sin \theta \; d \theta \; d\phi)$}, the proof in the appendix in \cite{Kay:1987ax} 
applies to the wave equation (\ref{rweqd}) for $ r^5 \dot Q^{>1}_-$. 
The intermediate case then follows, and so does the analogue  of Theorem 
1  in \cite{Kay:1987ax}. Note that the proof in   \cite{Kay:1987ax} shows that for the Klein Gordon field 
$|\Phi_{KG}| < C/r$ holds on the domain of outer communications. Although the weaker statement $|\Phi_{KG}| < C$ 
has been made in \cite{Kay:1987ax}, the stronger form was used in this proof. 
\end{proof}

\subsubsection{The asymptotically de Sitter case}

Two key similarities between  the extension (\ref{UVmetric})  for $\Lambda \neq 0$  and   the Kruskal extension of the  $\Lambda=0$ Schwarzschild black hole  (refer 
to Section \ref{bss} and Figure 1) are: i)  the $Z_2$ isometry 
exchanging I and $\rm{I}'$ and II and $\rm{II}'$, and ii) the fact that $r_h$ is a simple root of $f$, which implies that 
 the asymptotic behavior of fields vanishing at the bifurcation 
sphere is that in  equation (A1)   in  \cite{Kay:1987ax}. This allows us to prove the following
\begin{thm} \label{odd+}
 For any smooth solution of the odd LEE which has  compact support on Cauchy surfaces of the extended 
$ \rm{I \cup II \cup I' \cup II'}$ 
 SdS black hole, there exists a constant $K_-$ such that $|G_- | < K_- \; r^{-6}$ 
(equivalently, $|G_-|<$ a constant) for $r_h<r<r_c$.
\end{thm}
\begin{proof}
As in the proof of Theorem \ref{odd-0}, and in view of the above comments, we need only prove the intermediate case 
for $r^5 \dot Q_{>1}$. This field obeys the equation
\begin{equation} \label{2+2+L}
\frac{\p^2}{\p t^2} \left( r^6 \dot Q_{>1} \right) + \mathcal{A}^{\Lambda} \left( r^6 \dot Q_{>1} \right) =0,
\end{equation}
where (see (\ref{AL})) 
\begin{equation} \label{A+L}
\mathcal{A}^{\Lambda} = -\frac{\p^2}{\p {r^*}^2} + \left( 1 - \frac{2M}{r} - \frac{2 \Lambda}{3} r^2 \right) \left( -\frac{6M}{r^3} - \frac{\hd^A \hd_A}{r^2} \right) 
=: -\frac{\p^2}{\p {r^*}^2} + V_1^{\Lambda} - V_2^{\Lambda} \hd^A \hd_A, 
\end{equation}
with $-\infty < r^* < \infty$, 
\begin{equation}
V_1^{\Lambda} = \left( 1-\frac{2M}{r}-\frac{2 \Lambda}{3}  r^2  \right) \left( -\frac{6M}{r^3}\right)\;\; \text{ and  }\;\;
V_2^{\Lambda} = \left( 1-\frac{2M}{r}-\frac{2 \Lambda}{3}  r^2  \right) \left( \frac{1}{r^2} \right).   
\end{equation}
In view of  (\ref{rhrc}) and the condition $\ell \geq 2$, $V_1^{\Lambda}$ and $V_2^{\Lambda}$ are bounded in region II ($r_h<r<r_c$) and 
$ V_1^{\Lambda} - V_2^{\Lambda} \hd^A \hd_A >0$, thus 
$\mathcal{A}^{\Lambda}$ is a positive definite self adjoint operator on $L^2(\mathbb{R} \times S^2, dr^*  \sin \theta \; d \theta \; d\phi)$, and the proof 
follows as in Theorem~\ref{odd-0}.
\end{proof}

\subsubsection{A comment on the asymptotically anti de Sitter case} \label{comment1}

The stability proofs above use the facts that the operators (\ref{A}) and (\ref{A+L}) are self adjoint and positive definite in the region of interest. 
For negative cosmological constant,  the operator acts as 
 $(\ref{A+L})$  but on functions defined for $-\infty < r^* < 0$, which is the outer static region in this case.
This operator  is only {\em formally} self adjoint (i.e., if we ignore 
  boundary terms when integrating by parts), 
 we need to specify boundary conditions at $r^*=0$ to properly define a domain where 
the operator is self adjoint. 
There are different options, and the positivity or not of the resulting operator depends on the chosen boundary 
condition \cite{bernardo}. The stability  for Dirichlet boundary conditions was established in \cite{Ishibashi:2003ap}.

\section{The Linearized Einstein equation: even sector}
\subsection{Gauge transformations and gauge invariants}
The effect of the gauge transformation (\ref{gt})-(\ref{gte}) in the even sector (\ref{pert+}) is 
\begin{equation} \label{pert+gauge}
 {h'}_{\a \b}^{(+)} = \left( \begin{array}{cc}  
h_{ab}+ \td_a \xi_b + \td_b \xi_a & \hd_B (q_a + r^2 \td_a X + \xi_a^{>0})  \\ \hd_A (q_b + r^2 \td_b X + \xi_b^{>0}) & \;\;\;r^2 \left[\tfrac{1}{2} (J+ 2 \hd^C \hd_C X 
+ \frac{4}{r} \xi^a \td_a r) \widehat g_{AB} + 
(2 \hd_A \hd_B - \widehat g_{AB}\hd^C \hd_C) (G+X^{>1}) \right] \end{array} \right).
\end{equation}
Since $X$ and $q_a$ belong to  $L^2(S^2)_{>0}$, and $G\in L^2(S^2)_{>1}$,  the $\ell=0,1$ modes 
require a separate treatment.
\subsubsection{$\ell=0$ mode}
For the $\ell=0$ mode 
\begin{equation} \label{pert+gauge+0}
 {h'}_{\a \b}^{(\ell=0,+)} = \left( \begin{array}{cc}  
h_{ab}^{(\ell=0)} + \td_a \xi_b^{(\ell=0)}  + \td_b \xi_a^{(\ell=0)}  &0\\ 0 & \widehat  g_{AB}   (J^{(\ell=0)} 
 + \frac{4}{r} \xi^a_{(\ell=0)}  \td_a r)  \tfrac{r^2}{2} \end{array} \right), 
\end{equation}
we {\em partially}  fix a gauge by requiring ${J'}^{(\ell=0)} =0$ 
together with  the transverse condition $\tilde g^{ab} {h'}_{ab}^{(\ell=0)}=0$ \cite{Sarbach:2001qq} (note that a further gauge transformation with a gauge field 
\begin{equation} \label{res0}
\xi_a^{(\ell=0)} = \tilde \e_{a b} \td^b Z(r),
\end{equation}
would preserve these two conditions.)
 Dropping the primes, the perturbation in such a traceless gauge reads 
\begin{equation} \label{pert+T}
 h_{\a \b}^{(\ell=0,+,T)} = \left( \begin{array}{cc}  
h_{ab}^{(\ell=0,T)}  &0\\ 0 & 0 \end{array} \right), \;\;\;  \tilde g^{ab} h_{ab}^{(\ell=0,T)} =0.
\end{equation}
For the  traceless symmetric orbit space  tensor $h_{a b}^{(\ell=0,T)}$ we use the identity \cite{Sarbach:2001qq}
\begin{equation}\label{ttos}
{h}_{ab}^{(\ell=0,T)} = \frac{1}{f} \left[  C_a^{(\ell=0)} \td_b r + C_b^{(\ell=0)} \td_a r - \tilde g_{ab} C_d^{(\ell=0)} \td^d r \right], \;\;\;
 C_a^{(\ell=0)} = {h}_{ab}^{(\ell=0,T)} \td^b r,
\end{equation}
in terms of which, the  residual gauge transformation along the field  (\ref{res0}) gives 
\begin{equation} \label{gt0z}
C_a^{(\ell=0)} \to {C'}^{(\ell=0)}_a = C_a^{(\ell=0)} + f Z''(r) \tilde \e_{a b} \td^b r.
\end{equation}
In conclusion,  we  assume the form (\ref{pert+T})-(\ref{ttos}) for the metric perturbation, where $C_a^{(\ell=0)}$ is equivalent to  ${C'}^{(\ell=0)}_a$ 
defined in (\ref{gt0z}).

\subsubsection{$\ell=1$ modes}
For  $\ell=1$ modes
\begin{equation} \label{pert+gauge+1}
 {h'}_{\a \b}^{(\ell=1,+)} = \left( \begin{array}{cc}  
h_{ab}^{(\ell=1)} + \td_a \xi_b^{(\ell=1)}  + \td_b \xi_a^{(\ell=1)}  & \hd_B (q_a^{(\ell=1)} + r^2 \td_a X^{(\ell=1)} + \xi_a^{(\ell=1)})  \\ 
\hd_A (q_b^{(\ell=1)} + r^2 \td_b X^{(\ell=1)} + \xi_b^{(\ell=1)})   &\;\;\;\; \widehat  g_{AB}   (J^{(\ell=1)} 
- 4 X^{(\ell=1)}  + \frac{4}{r} \xi^a_{(\ell=1)}  \td_a r ) \frac{r^2}{2} \end{array} \right).
\end{equation}
We may choose the gauge field such that $ r^2 \td_a X^{(\ell=1)} + \xi_a^{(\ell=1)} = -q_a^{(\ell=1)}$ and 
$2 \td^b \xi_b^{(\ell=1)}= -\tilde g^{ab}h_{ab}^{(\ell=1)}$. This will partially fix an orbit space
 transverse gauge, leaving a perturbation of the form
\begin{equation} \label{pert+RW}
 h_{\a \b}^{(\ell=1,+,T)} = \left( \begin{array}{cc}  
h_{ab}^{(\ell=1,T)}  &0\\ 0 &  \frac{r^2}{2}  \widehat  g_{AB}   J^{(\ell=1)}  \end{array} \right), \;\;\;  \tilde g^{ab} h_{ab}^{(\ell=1,T)} =0.
\end{equation}
where, as in the $\ell=0$ sector, 
\begin{equation}\label{tt1s}
{h}_{ab}^{(\ell=1,T)} = \frac{1}{f} \left[  C_a^{(\ell=1)} \td_b r + C_b^{(\ell=1)} \td_a r - \tilde g_{ab} C_d^{(\ell=1)} \td^d r \right], \;\;\;
 C_a^{(\ell=1)} = {h}_{ab}^{(\ell=1,T)} \td^b r,
\end{equation}
A residual gauge transformation along a   field (\ref{vec+}) satisfying 
\begin{equation} \label{rgf1}
 r^2 \td_a X^{(\ell=1)} + \xi_a^{(\ell=1)} =0 , \;\; \td^a  \xi_a^{(\ell=1)} =0,
\end{equation}
preserves the form (\ref{pert+RW})-(\ref{tt1s}).

\subsubsection{$\ell > 1$ modes}
For $\ell > 1$, equation (\ref{pert+gauge})
\begin{align} \begin{split}
{h'}_{ab}^{>1} &= h_{ab}^{>1} + \td_a \xi_b^{>1} + \td_b \xi_a^{>1}\\
{q'}_a^{>1} &= q_a^{>1} + r^2 \td_a X^{>1} +\xi_a^{>1}\\
{J'}^{>1} &= J+ 2 \hd^C \hd_C X^{>1} + \frac{4}{r} \xi_a^{>1} \td^a r\\
G' &= G + X^{>1},
\end{split} \label{gauge+}
\end{align}
has two implications (compare with the discussion following (\ref{pert-gt})): i) the fields
\begin{align} \label{gief}
 H_{ab} &= h_{ab}^{\ell>1} - \td_a p_b -\td_b p_a, \;\;\; (p_a \equiv q_a^{>1} -r^2 \td_a G)\\
\mathcal{J} &= J - \frac{4}{r} p_a \td^a r - 2 \hd^C \hd_C G,
\end{align}
are gauge invariants; and (ii) there is a (unique) gauge (the Regge-Wheeler gauge) where the $\ell>1$ piece of the even metric 
perturbation assumes the form
\begin{equation} \label{RWe}
 {h}_{\a \b}^{(+,>1)} = \left( \begin{array}{cc}  
H_{ab}  & 0 \\ 
0  &  \frac{r^2 }{2} \; \widehat  g_{AB} \,  \mathcal{J}  \end{array} \right).
\end{equation}

\subsection{The linearized Einstein equation}

\subsubsection{$\ell>1$ modes}
A calculation of the LEE (\ref{lee}) using (\ref{RWe}) gives, in agreement with  equations (31) in \cite{Chaverra:2012bh}, the 
following components for $\dot R_{\a \b} - \Lambda h_{\a \b}$:
\begin{align}  \nonumber
\dot R_{ab} - \Lambda h_{a b} &= \frac{\td^cr}{r} \left(\td_a H_{bc} +\td_b H_{ac}-\td_cH_{ab} \right) -\frac{1}{2r^2} \hd^C \hd_C H_{ab} 
+ \frac{\tilde R}{2} H_{ab}\\
& \;\;\;\; +\frac{1}{2} \tilde g_{ab} \left( \td^c \td^d H_{cd} - \td^c \td_c H - \frac{\tilde R}{2} H \right) - \frac{1}{2r^2} \td_{(a}\left( r^2 \td_{b)} 
\mathcal{J} \right)- \Lambda H_{ab},   \label{even1} \\  \label{even2}
\dot R_{aB} - \Lambda h_{a B} &= \frac{1}{2} \hd_B \left[ \td^b H_{ab}-r \td_a \left( \frac{H}{r} \right) - \frac{1}{2} \td_a \mathcal{J} \right]\\
\dot R_{AB} - \Lambda h_{A B} &= -\frac{1}{2} \left( \hd_A \hd_B  H - \tfrac{1}{2} \widehat g_{AB} \hd^C \hd_C H \right) 
 + \widehat g_{AB} \left[ \td^a (r H_{ab} \td^b r)  \nonumber \right. 
\\ & \;\;\;\; \left. -\frac{r}{2} (\td^a r) \td_a H - \frac{1}{4} \td^c \td_c (r^2 \mathcal{J}) - \frac{1}{4} 
\hd^C \hd_C(H+ \mathcal{J}) - \Lambda \frac{r^2}{2} \mathcal{J} \right], \label{even3}
\end{align}
where $\tilde R$ is the Ricci scalar of the orbit manifold and 
\begin{equation} \label{f2}
H = \tilde g^{a b} H_{a b}.
\end{equation}
From part ii) of Lemma \ref{kernel} (applied now to the symmetric tensor $\dot R_{\a \b} - \Lambda h_{\a \b}$)  we conclude that (\ref{even3}) 
implies $H=0$, then as in (\ref{ttos}) we introduce 
\begin{equation} \label{C>1}
 C_a^{>1} = {H}_{ab} \td^b r,
\end{equation}
which gives 
\begin{equation}\label{ttlg1}
{H}_{ab} = \frac{1}{f} \left[  C_a^{>1} \td_b r + C_b^{>1} \td_a r - \tilde g_{ab} C_d^{>1} \td^d r \right].
\end{equation}
and
\begin{equation}
\tilde \e^{ac} (\td_c r) \td^b H_{ab} = \tilde \e^{ac} \td_c C_a.
\end{equation}
For $\ell>0$, equation  (\ref{even2}) is equivalent to $\td^b H_{ab}- \frac{1}{2} \td_a \mathcal{J}=0$. Contracting this equation 
 with the orthogonal vectors
$\td^a r$ and $\tilde  \e^{a b} \td_b r$   and using (\ref{rab}) gives \cite{Chaverra:2012bh}:
\begin{align} \nonumber
0 &= \td^a r \left( \td^b H_{ab}- \frac{1}{2} \td_a \mathcal{J} \right) = \td^b \left( H_{ab} \td^ar \right) - H_{ab} \td^b \td^a r -  \frac{1}{2} (\td^a r) 
(\td_a \mathcal{J}) \\ &= \td^b C_b^{>1} -\frac{1}{2}  (\td^b r) (\td_b \mathcal{J}),   \label{eqc} \\ \nonumber
0 &= \tilde \e^{ac} (\td_c r) \left( \td^b H_{ab}- \frac{1}{2} \td_a \mathcal{J} \right) \\ &=
\tilde \e^{a c} \left( \td_c C_a^{>1} - \frac{1}{2} (\td_c r) \td_a \mathcal{J} \right). \label{dz1}
\end{align}
Introducing 
\begin{equation} \label{Za}
Z_a \equiv C_a^{>1} - \frac{r}{2} \td_a \mathcal{J},
\end{equation}
we write (\ref{dz1}) as
\begin{equation}  \label{41a}
\td_{[a} Z_{b]} =0
\end{equation}
and equation (\ref{eqc}) as
\begin{equation} \label{41b}
\td^b Z_b + \frac{r}{2} \td^b \td_b \mathcal{J} =0.
\end{equation} 
Now contract (\ref{even3}) with $\widehat g^{AB}$, this gives 
\begin{equation} \label{41c}
4 \td^a \left( rZ_a \right) + r^2 \td^c \td_c \mathcal{J} - \left[ \hd^C \hd_C +2 \right] \mathcal{J} =0.
\end{equation}
Finally, contracting the $\tilde g_{ab}$ trace-free part of (\ref{even1}) with $\td^a r$ and using (\ref{41c}) we arrive at
\begin{equation} \label{41d}
\td_a \left[ 2 r (\td^b r) Z_b + (3M-r) J -  \frac{r}{2} \hd^C \hd_C  J  \right] - \hd^C \hd_C Z_a =0
\end{equation}
It is interesting to note that
 equations (\ref{41a})-(\ref{41d}) look {\em formally} identical to the $\Lambda=0$ case, equations (41a)-(41d) in  \cite{Chaverra:2012bh},
 $\Lambda$ appears only implicitly through $\tilde g_{ab}$ and its Levi-Civita  derivative $\td_a$.

\subsubsection{$\ell=0$ mode}

In deriving equations (\ref{even1})-(\ref{even3}) we made no assumptions on $H_{ab}$ and $\mathcal{J}$  in (\ref{pert-RW}),  so these calculations
 apply  to the $\ell=0$ mode (\ref{pert+T}) if we replace $H_{ab}$ with $h_{ab}^{(\ell=0,T)}$ and set $\mathcal{J}=0$. 
Equation (\ref{even2}) is void in this case, whereas equation (\ref{even3})  reduces to $\td^a (rC_a^{(\ell=0)})=0$, whose solution is 
\begin{equation} \label{ca0}
C_a^{(\ell=0)} =  \frac{1}{r} \tilde \e_{a b} \td^b z, \;\; z: \mathcal{O} \to \mathbb{R}. 
\end{equation}
In Schwarschild coordinates $(t,r)$ the residual gauge freedom  (\ref{gt0z}) implies that  $z(t,r)$ is defined up to an arbitrary additive function $x(r)$ 
(choose $x'=rf Z''$ to 
match (\ref{gt0z})). 
Replacing (\ref{ca0})  in  (\ref{even1}) and (\ref{even3}) we find that,working in $(t,r)$ coordinates,   $z(t,r)=At + B(r)$. This gives 
$h_{tt}= A/r$, $h_{rr}= (2A/r)/f^2$ and $h_{tr}=B'(r)/r$. Choosing the gauge $B(r)=0$ we recognize this perturbation as 
a  shift $M \to M +A/2$ in the mass treated to first order in $A/2 = \dot M$. This was to be expected  from Birkhoff's theorem. 
In conclusion, we can choose a gauge such that 
\begin{equation} \label{md}
C_a^{(\ell=0)} = \frac{2\;  \dot M}{rf}\;  \td_a r.
\end{equation}
The perturbation class is characterized by the  parameter $\dot M$.

\subsubsection{$\ell=1$ modes}

We use again equations (\ref{even1})-(\ref{even3}) with the replacements $H_{ab} \to h_{ab}^{(\ell=1,T)}$ and $J \to J^{(\ell=1)}$, 
and find that the general solution to these equations can be set to zero using the residual gauge freedom (\ref{rgf1}) (see \cite{Sarbach:2001qq}). 
This implies that the even $\ell=1$ sector is void. 

\subsection{Solution of the linearized Einstein equation}

The results of the previous section are gathered in the following

\begin{lem}
The  LEE (\ref{lee}) in the even sector is equivalent  to the system of equations  (\ref{Za})-(\ref{41d}) and (\ref{md}).
\end{lem}

The system (\ref{41a})-(\ref{41d}) of LEE for the $\ell>1$ sector was first solved by Zerilli in \cite{Zerilli:1970se}, the addition of a cosmological constant 
 was considered in \cite{guven}. In this section we generalize to the case  $\Lambda \neq 0$ the derivation in \cite{Chaverra:2012bh} of the Zerilli equation.\\
 Equation  (\ref{41a}) implies that there is a scalar field $\zeta: \mathcal{M} \to \mathbb{R}$, defined up to an  additive 
function $z: S^2 \to \mathbb{R}$, 
$\zeta = \zeta_o + z$, such that 
\begin{equation} \label{Zz}
Z_a = \td_a \zeta.
\end{equation}
From (\ref{41b}) and (\ref{41c})
\begin{equation} \label{42}
\td^b Z_b + \frac{2 \td^b r}{r} Z_b - \frac{(\hd^C \hd_C+2)}{2r} \mathcal{J}=0.
\end{equation}
From (\ref{41d}) (\ref{Zz}) and the above equation
\begin{equation} \label{43p}
\td^b \td_b \zeta_o + \frac{\hd^C \hd_C}{r^2} \zeta_o - \frac{3M}{r^2} \mathcal{J}=\frac{h-\hd^C \hd_C z}{r^2},  \;\;\; h: S^2 \to \mathbb{R}.
\end{equation}
Since all fields above belong to $L^2(S^2)_{>1}$, we may choose $z$ such that the term on the right vanishes. 
From now on we assume this choice, which gives 
\begin{equation} \label{43}
\td^b \td_b \zeta + \frac{\hd^C \hd_C}{r^2} \zeta - \frac{3M}{r^2} \mathcal{J}=0.
\end{equation}
Applying $(\hd^C \hd_C+2)$ to (\ref{43})  and combining with (\ref{42}) yields 
\begin{equation} \label{zeq}
\left(\hd^C \hd_C + 2 - \frac{6M}{r}  \right)  \td^b \td_b \zeta - \frac{12M}{r^2}  \td^b r \td_b \zeta + \frac{(\hd^B \hd_B)(\hd^C \hd_C + 2)}{r^2} \zeta=0.
\end{equation}
The Zerilli field
\begin{equation} \label{preZ}
\Psi = \left(\hd^C \hd_C + 2 - \frac{6M}{r}  \right)^{-1} \zeta
\end{equation}
is introduced to eliminate first derivatives from (\ref{zeq}) and reduce it to a two dimensional wave equation:
\begin{multline} \label{zerilli}
\td^b \td_b \; \Psi + \frac{1}{r^2} \left( \hd^B \hd_B + 2 - \frac{6M}{r} \right)^{-2}  \left[  (\hd^B \hd_B +2)^2 \left(\hd^B \hd_B -\frac{6M}{r} \right) 
\right. \\ \left. +\frac{36 M^2}{r^2} \left( \hd^B \hd_B +2-\frac{2M}{r} + \frac{2}{3} \Lambda r^2 \right) \right] \Psi=0.
\end{multline}
This is the Zerilli equation, first obtained for $\Lambda=0$ in \cite{Zerilli:1970se}, and generalized to Schwarzschild-(A)dS   in  \cite{guven}. 
The non-local operator $( \hd^B \hd_B + 2 - 6M/r)^{-2}$  is well  defined on $L^2(S^2)_{>1}$;  
the Zerilli equation 
was derived within the context of a  modal approach to the problem, in which $\Psi$ is expanded 
in spherical harmonics 
\begin{equation} \label{zerilli-exp}
\Psi = \sum_{(\ell,m)} \phi^+_{(\ell,m)} S_{(\ell,m)},
\end{equation}
and this operator  reduces to 
$\left( (\ell+2)(\ell-1) - \frac{6M}{r} \right)^{-2}$ in the $\ell$ subspace,  and it  is therefore suitable to 
perform explicit calculations. If we use a tortoise radial coordinate $r^*$, we find that (\ref{zerilli}) is equivalent to  the standard
 form (\ref{ZE})-(\ref{zp}) of  the equation in the original references \cite{Zerilli:1970se}  \cite{guven}. 
From the $\zeta$ field (and therefore from the Zerilli field) 
 it is possible to reconstruct $H_{ab}$ and $\mathcal{J}$ by tracing back the above equations. 
The result is  \cite{Buchman:2007pj} \cite{Chaverra:2012bh}:
\begin{align} \label{J}
J &= 2 \left( \hd^C \hd_C + 2 - \frac{6M}{r}  \right)^{-1}  \left[ 2 \td^b r \td_b \zeta - \frac{1}{r} \hd^B \hd_B \zeta \right] \\
H_{ab} &= 2 \left( \hd^C \hd_C + 2 - \frac{6M}{r}  \right)^{-1}  \left[ \td_a \td_b (r\zeta) -\frac{\tilde g_{ab}}{2} 
 \td^c \td_c (r\zeta) \right] \label{H} 
\end{align}
The even sector LEE is equivalent to the set (\ref{zeq}), (\ref{J})-(\ref{H}).

\subsection{The ubiquitous Regge-Wheeler equation}

Let us consider the field \cite{Chaverra:2012bh} 
\begin{equation} \label{fifromz}
\Phi = -r \td^b \td_b \zeta.
\end{equation}
Applying $\td^a \td_a$ to $r^2$ times equation (\ref{43}), and  eliminating $\td^b \td_b \mathcal{J}$ 
using  (\ref{41b}), we find that $\Phi$ satisfies
\begin{equation} \label{rw3}
\td^c \td_c \Phi + \frac{2}{r}\td^c r \td_c \Phi + \frac{1}{r^2} \hd_A \hd^A \Phi + \left( \frac{6M}{r^3} +\frac{1}{r} \td^c \td_c r \right) \Phi =0,
\end{equation}
which is   the Regge-Wheeler equation (\ref{rweP})! We have found a scalar field $\Phi$, related to the even metric 
perturbation potential $\zeta$ through (\ref{fifromz}), that satisfies the fundamental equation to which the {\em odd} 
LEE reduces.\\
Replacing $\td^b \td_b \zeta$ with $-\Phi / r$ in (\ref{zeq}) gives the following   relation between $\zeta$ and 
$\Phi$:
\begin{equation} \label{50}
\left[12M  \; (\td^a r) \; \td_a  - (\hd^C \hd_C + 2) ( \hd^C \hd_C) \right] \zeta
= \left[ \frac{6M}{r}  - (\hd^C \hd_C + 2)  \right] (r \Phi).
 \end{equation}
Since $\zeta$ (or $\Psi$) contains all the information on equivalence classes of  $\ell >1$ solutions of  the even LEE, but does  not  
admit a  four dimensional translation, whereas $\Phi$ satisfies the 4DRW equation,  one is tempted to treat even perturbations  
in terms of  $\Phi$, using the relationship (\ref{50}). This possibility was disregarded in  \cite{Chaverra:2012bh}  due to 
the fact that the operator on the left hand side in (\ref{50}) has a non trivial kernel, suggesting that 
information is lost when switching from $\zeta$ to $\Phi$. There is, however, a loophole in this argument, as we now proceed 
to explain: \\
Expand $\Phi$ and $\zeta$ in spherical harmonics (we choose to call $\phi_{(\ell,m)}^{(-)}$ the components of 
$\Phi$ in (\ref{fifromz}) since this field satisfies (\ref{4DRWE}), therefore the components (\ref{4DRf}) satisfy (\ref{RWE}))
\begin{equation}\label{shd}
\Phi = \sum_{\ell \geq 2, m} \frac{\phi^{(-)}_{(\ell,m)}}{r} S_{(\ell,m)}, \;\;\;
 \zeta = \sum_{\ell \geq 2, m} \zeta_{(\ell,m)} S_{(\ell,m)} = 
- \sum_{\ell \geq 2, m} \left[[ (\ell+2)(\ell-1) + \frac{6M}{r} \right] \phi^+_{(\ell,m)} S_{(\ell,m)}
\end{equation}
 In $(t,r)$ coordinates,  equation (\ref{50}) with the replacements (\ref{shd}) reads
\begin{equation} \label{dual2}
\left[ f \frac{\p}{\p r} - w_{\ell} \right] \zeta_{(\ell,m)} =  \left[  \frac{1}{2r}  + \frac{(\ell+2)(\ell-1)}{12M}  \right] \phi^{(-)}_{(\ell,m)},
\end{equation}
where 
\begin{equation} \label{wl}
w_{\ell} = \frac{1}{12M} \frac{(\ell+2)!}{(\ell-2)!}
\end{equation}
are the frequencies of the Chandrasekhar algebraically special modes \cite{ch1} \cite{ch2}. 
 The general solution of (\ref{dual2}) can be written as 
\begin{equation} \label{gsz}
\zeta_{(\ell,m)}(t,r)  = F_{r_o}(t) \; e^{w_{\ell} r^*} - e^{w_{\ell} r^*}\;
\int_{r^*}^{r_o^*}  d{r^*}' \;  e^{-w_{\ell} {r^*}'} \left[ \phi^{(-)}_{(\ell,m)}(t,r) \;  \left(  \frac{1}{2r({r^*}')}  + \frac{(\ell+2)(\ell-1)}{12M}  \right) 
\right]_{r=r({r^*}')}
\end{equation}
where $r(r^*)$ is the inverse of the function $r^*(r)$ defined in (\ref{rs}), and $r_o$ is the radial coordinate of a point in the outer 
static region. The non-trivial kernel of the operator on the left of equation (\ref{50}) is the reason why 
there is an arbitrary function of $t$ in the first term in (\ref{z}). Note that 
\begin{equation}
 F_{r_o}(t) = \zeta_{(\ell,m)}(t,r_o)  \;\; e^{-w_{\ell} r_o^*} 
\end{equation}
If we apply $\td^b \td_b$ to $\zeta_{(\ell,m)}$ in (\ref{gsz}) and use (\ref{RWE}) we find that (compare with 
equations (\ref{fifromz}) and (\ref{shd}))
\begin{equation} \label{bxz=phi}
e^{-w_{\ell} r^*} f \left[\td^b \td_b \zeta_{(\ell,m)} + r^{-2} \phi^{(-)}_{(\ell,m)} \right] =
(w_{\ell}^2 F_{r_o} - \ddot F_{r_o} )  +  q_{r_o}(t),
\end{equation}
where
\begin{multline}
q_{r_o}(t)= 
\left[\left( \frac{w_{\ell} (\ell+2)(\ell-1)}{12M} + \frac{w_{\ell}}{2 r_o}+ \frac{1}{2 {r_o}^2}-\frac{M}{ {r_o}^3}\right)
 \phi_{(\ell,m)}^{(-)}(t,r_o)  \right. \\ \left.
+ \left( \frac{(\ell+2)(\ell-1)}{12M}- \frac{\ell^2+\ell-5}{6 r_o}-\frac{M}{{r_o}^2} \right) \p_r \phi^{(-)}_{(\ell,m)}(t,r_o) \right] 
e^{-w_{\ell}r^*_o}
\end{multline}
Let us consider equations (\ref{gsz})-(\ref{bxz=phi}) for different values of the cosmological constant:

\begin{itemize}
\item  $\Lambda=0$: in the asymptotically flat case (\ref{rs}) and (\ref{rs0}) in (\ref{gsz}) gives 
\begin{multline} \label{z}
\zeta_{(\ell,m)}(t,r)  = F_{r_o}(t) \; e^{w_{\ell} r} \left(\frac{r}{2M}-1\right)^{2Mw_{\ell}} \\ - \frac{e^{w_{\ell} r}}{2M}  \left(\frac{r}{2M}-1\right)^{2Mw_{\ell}} 
\int_r^{r_o}   r' 
e^{-w_{\ell} r'} {\phi^{(-)}_{(\ell,m)}}_{\big|_{(t,r')}} \left(  \frac{1}{2r'}  + \frac{(\ell+2)(\ell-1)}{12M}  \right)  \left(\frac{r'}{2M}-1\right)^{-2Mw_{\ell}-1} dr',
\end{multline}
for some $r_0>2M$. 

 According to Theorem \ref{odd-0},  $\Phi$  in (\ref{50}), being a solution of the $\Lambda=0$  (\ref{4DRWE}), 
 satisfies $|\Phi| < C/r$ on the exterior Kruskal wedge, then from (\ref{shd})  $\phi_{(\ell,m)}<$ constant
\footnote{The boundedness of the $\phi^+_{(\ell,m)}$  and  the  for $r>2M$ and all $t$ was first established 
in \cite{wald}.}. More generally, solutions of (\ref{RWE}) behave either as  $r^{\ell+1}$ or $r^{-\ell}$ for large $r$, and the first 
type should be discarded to preserve asymptotic flatness (alternatively, to assure the perturbative character of the initial datum).
In any case, 
the integral in (\ref{z}) converges if  we  take $r_o = \infty$  in (\ref{z}), 
\begin{multline} \label{z2}
\zeta_{(\ell,m)}(t,r)  = F_{\infty}(t) \; e^{w_{\ell} r} \left(\frac{r}{2M}-1\right)^{2Mw_{\ell}} \\ - \frac{e^{w_{\ell} r}}{2M}  
\left(\frac{r}{2M}-1\right)^{2Mw_{\ell}} 
\int_r^{\infty}   r' e^{-w_{\ell} r'} {\phi^{(-)}_{(\ell,m)}}_{\big|_{(t,r')}} \left(  \frac{1}{2r'}  + \frac{(\ell+2)(\ell-1)}{12M}  \right)  
\left(\frac{r'}{2M}-1\right)^{-2Mw_{\ell}-1} dr',
\end{multline}
and  the second term above  is   bounded for fixed $t$ as $r \to \infty$. Therefore, for  metric perturbations that do not diverge 
as $r\to \infty$, it must be  $F_{\infty}(t)=0$, otherwise the $\ell$ pieces of 
$J$ and $H_{ab}$ would behave for large $r$  as $e^{w_{\ell} r}$ times a rational function of $r$ (see (\ref{J})-(\ref{H})).
 We conclude  that the asymptotic condition  as $r\to \infty$ resolves the ambiguity in (\ref{50}) and yields a 1-1 relation 
between the $\zeta$ and $\Phi$ fields:
\begin{equation} \label{z3}
\zeta_{(\ell,m)}(t,r)  = - \frac{e^{w_{\ell} r}}{2M}  \left(\frac{r}{2M}-1\right)^{2Mw_{\ell}} 
\int_r^{\infty} r' e^{-w_{\ell} r'} {\phi^{(-)}_{(\ell,m)}}_{\big|_{(t,r')}} \left(  \frac{1}{2r'}  + \frac{(\ell+2)(\ell-1)}{12M}  \right)  \left(\frac{r'}{2M}-1\right)^{-2Mw_{\ell}-1} dr'
\end{equation}

Had we used $r_o < \infty$ in (\ref{z}), the non-trivial function $F_{r_o}(t)$ could have been obtained by requiring that (\ref{z}) be 
a solution of  (\ref{fifromz}) 
(or, equivalently, of equation (\ref{zeq})) for 
  $\phi_{(\ell,m)}$ satisfying   (\ref{RWE}). As follows from  equation (\ref{bxz=phi}),  this implies that $F_{r_o}(t)$ must be  
 a solution of the ordinary differential equation 
resulting by setting the right hand side of this equation  equal to zero. The ambiguity $A e^{w_{\ell}t}+ B e^{-w_{\ell}t}$ in the solution 
of this  equation is again fixed by adjusting $A$ and $B$ such that (\ref{gsz}) remains bounded as $r \to \infty$ (this gives $A=B=0$ 
when $r_o=\infty$, in agreement with our previous paragraph). \\

 Back to the case $r_o=\infty$, if we  we let  $F_{\infty} \neq 0$ in (\ref{z2}), allowing   
perturbations that diverge  as $r \to \infty$, but require that 
\begin{equation}
(\p_t^2 - \p_{r^*}^2) \left[F_{\infty}(t) \; e^{w_{\ell} r} \left(\frac{r}{2M}-1\right)^{2Mw_{\ell}} \right]= (\p_t^2 - \p_{r^*}^2) [F_{\infty}(t) 
\;\exp({w_{\ell}}r^{*}) ] =0 
\end{equation}
to preserve (\ref{fifromz}), we obtain $\zeta_{\infty} \equiv F_{\infty}(t) \;\exp({w_{\ell}}r^{*}) = \exp( w_{\ell} (r^* \pm t))$ and, from (\ref{preZ}),  
the following solutions to the Zerilli equation
\begin{equation} \label{ass}
\phi^{(+)}_{(\ell,m)} = \frac{r \exp( w_{\ell} (r^* \pm t))}{(\ell+2)(\ell-1)r+6M}.
\end{equation}
The relevance of (\ref{ass}) (equivalently, the solutions
 $\zeta= \exp( w_{\ell} (r^* \pm t))$ of (\ref{zeq})) comes from the fact that 
these are solutions of the Zerilli equation (respectively (\ref{zeq})) for 
  {\em any} $\Lambda$ (assuming  the appropriate $r^*(r)$ satisfying (\ref{rs}) is used) and, 
although useless when $\Lambda \geq 0$ for their behavior for large $r^*$, they  are  valid for $\Lambda <0$  
 since $r^*<0$ in this case,  and show that there are unstable solutions in the asymptotically AdS case. This is discussed in 
detail in section \ref{ins-neg-lam}.

\item $\Lambda >0$:  
in this case   $r_o \in (r_h, r_c)$ in (\ref{gsz}). 
According to Theorem \ref{odd+},  since $\Phi$  in (\ref{50}) is a solution of the $\Lambda>0$  (\ref{4DRWE}), 
the $\phi_{(\ell,m)}$ in (\ref{dual2}) are bounded for $r_h < r < r_c$ ($-\infty < r^* < \infty$). 
This implies that we can take $r_o \to r_c$ ($r_o^* \to \infty$) 
in (\ref{gsz}), and  the resulting term involving the integral 
will be bounded for fixed $t$ as $r \to r_c$. Therefore, as in the previous case, we conclude that for   metric perturbations that do not diverge 
as $r\to r_c$, the only consistent  choice  is $F_{r_c}(t)=0$. Otherwise the $\ell$ pieces of 
$J$ and $H_{ab}$ would diverge as $e^{w_{\ell} r^*}$ as $r \to r_c$. 
Once again, the asymptotic condition   resolves the ambiguity in (\ref{50}) and yields a 1-1 relation 
between the $\zeta$ and $\Phi$ fields:
\begin{equation} \label{z+}
\zeta_{(\ell,m)}(t,r)  = - e^{w_{\ell} r^*}
\int_{r^*}^{\infty}     d{r^*}' \; e^{-w_{\ell} {r^*}'} \left[ \phi^{(-)}_{(\ell,m)}(t,r) \left(  \frac{1}{2r}  + \frac{(\ell+2)(\ell-1)r}{12M}  \right)\right]_{r=r({r^*}')}
\end{equation}

\item $\Lambda <0$:  For asymptotically AdS black holes, we cannot use the argument above, 
since we have proven no analogue of Theorems \ref{odd-0} 
and \ref{odd+} in this case. 
This is connected to the fact that there are different consistent choices for the behavior of the $\phi_{(\ell,m)}$ 
as $r \to \infty$ ($r^* \to 0^-$), and they lead to different dynamics \cite{bernardo}. As an example, a boundary condition 
consistent 
with (\ref{ass}) gives a dynamics under which the asymptotically AdS Schwarzschild black hole is unstable. The possibility of replacing the Zerilli equation 
(\ref{zerilli}) for the Regge-Wheeler equation (\ref{rw3}) depends on the choice of boundary conditions at $r^*=0$. This is further investigated in 
section \ref{chd} and in \cite{bernardo}.

\end{itemize}

We summarize below the results of this section:
\begin{lem} \label{evensol}
Consider the even sector of the LEE around a Schwarzschild-(A)dS background:
\begin{itemize}
\item[(i)] 
The solution of the $\ell=0$ piece of the metric perturbation 
in a particular gauge is given by equations (\ref{pert+T}), (\ref{ttos}) and  (\ref{md}). 
\item[(ii)] The $\ell=1$ sector of the LEE is trivial.
\item[(iii)] For $\Lambda \geq 0$, the solution of the $\ell>1$ sector  is given in equations (\ref{rweP}), (\ref{shd}), (\ref{z3})/(\ref{z+}) and 
 (\ref{J})-(\ref{H}).
\item[(iv)] For $\Lambda \geq 0$ there is a bijection between the  space $\mathcal{L}_+$ of smooth odd solutions of the LEE mod gauge transformations, 
and the set $$L_+ = \{ \dot M \} \cup \{ \Phi \; | \;  \Phi \text{ is a smooth solution of equation (\ref{rweP})} \}.$$ 
\item[(v)]  For $\Lambda \geq 0$ there is a  bijection between  $\mathcal{L}_+$ and 
the set $$L_+^{\phi} = \{ \dot M \} \cup \{ \phi^+_{(\ell,m)}, m, \ell\geq 2 \; | \;  \phi^+_{(\ell,m)}  \text{ is a smooth solution of equation (\ref{ZE})} \}.$$
\end{itemize}
\end{lem}

We gather   Lemma \ref{oddsol}.iii and Lemma \ref{evensol}.iv in   

\begin{thm} \label{bij}
For $\Lambda \geq 0$ there is a bijection between the  space $\mathcal{L}$ of smooth solutions of the LEE mod gauge transformations and 
the set 
\begin{equation} \label{param}
L = \{ \dot M,  j^{(m)}, m=1,2,3 \} \cup \{ (\Phi_-,\Phi_+) \; | \;  \Phi_{\pm} \text{ smooth solutions of equation (\ref{rweP})} \}.
\end{equation}
The dynamical perturbations are parametrized by the two solutions $\Phi_{\pm}$ of the 4DRW equation (\ref{rweP}), and correspond to 
$\ell \geq 2$ perturbations. The stationary perturbations 
are parametrized by the first order variation  of the mass  $\dot M$ ($\ell=0$) and the angular momenta components  $j^{(m)}, m=1,2,3$ ($\ell=1$), these correspond to perturbations within the Kerr/ Kerr-(A)dS family.
\end{thm}

\subsection{Chandrasekhar's duality} \label{chd}

The modal approach to the linear perturbation problem is based on analyzing  the evolution of isolated $(\ell,m)$ modes using the Zerilli and Regge-Wheeler 
equations (\ref{ZE}) and (\ref{RWE}) respectively. In $(t,r)$ coordinates, these equations are separable, there are solutions of the form 
\begin{equation}
\phi^+_{(\ell,m)} = \Re \; e^{i \w t} \psi^+_{(\ell,m)}(r), \;\;\; \phi^-_{(\ell,m)} = \Re \; e^{i \w t} \psi_{(\ell,m)}^-(r),
\end{equation}
where $\psi^{\pm}_{(\ell,m)}$ satisfies a  Schr\"odinger-like equation (note that $\mathcal{H}^{-}_{\ell}$ agrees 
with the operator  
$\mathcal{A}^{\Lambda}$ introduced in (\ref{AL}))
\begin{equation} \label{sle}
\mathcal{H}^{\pm}_{\ell} \psi^{\pm}_{(\ell,m)}  :=  [-\p_{r^*}^2  + U^{\pm}_{\ell}] \psi^{\pm}_{(\ell,m)} = \w^2 \psi^{\pm}_{(\ell,m)},
\end{equation}
with potentials (see (\ref{rwp}) (\ref{zp}))
\begin{equation}
 U^+_{\ell} = f V^Z_{\ell}, \;\;\; U^-_{\ell} = f V^{RW}_{\ell}.
\end{equation}

It was noticed by Chandrasekhar \cite{ch1} \cite{ch2} (see also Appendix A in \cite{Berti:2009kk}) 
 that  the $\mathcal{H}^{\pm}_{\ell} = -\p_{r^*}^2  + U^{\pm}_{\ell}$ 
satisfy
\begin{equation} \label{factor}
\mathcal{H}^{\pm}_{\ell} + {w_{\ell}}^2 =  \mathcal{D}^{\pm} \mathcal{D}^{\mp}, 
\end{equation}
where, generalizing Chandrasekhar's equations  to $\Lambda\neq 0$, 

\begin{align} \label{dpm}
\mathcal{D}^{\pm}_{\ell} &= \pm \frac{\p}{\p r^*} + W_{\ell}, \\ \label{dpm2}
W_{\ell} &= w_{\ell} + \frac{6M f}{r \; (\mu r + 6M)},
\end{align}
and $w_{\ell}$ are  the frequencies (\ref{wl}). 
A consequence of the factorization (\ref{factor}) is that
\begin{align} \label{iso1}
\mathcal{H}^{-}_{\ell} \psi^{-}_{(\ell,m)}  = \w^2 \psi^{-}_{(\ell,m)} 
&\Rightarrow \mathcal{H}^{+}_{\ell} (\mathcal{D}_{\ell}^+ \psi^{-}_{(\ell,m)})
  = \w^2  (\mathcal{D}^+_{\ell} \psi^{-}_{(\ell,m)})\\ \label{iso2}
\mathcal{H}^{+}_{\ell} \psi^{+}_{(\ell,m)}  = \w^2 \psi^{+}_{(\ell,m)} 
&\Rightarrow \mathcal{H}^{-}_{\ell} (\mathcal{D}^-_{\ell}  \psi^{+}_{(\ell,m)})  = \w^2  (\mathcal{D}^-_{\ell}
 \psi^{+}_{(\ell,m)})
\end{align}
Similarly, solutions of the (\ref{RWE}) and (\ref{ZE}) are exchanged by $\mathcal{D}^{\pm}_{\ell}$, e.g,
\begin{align} \label{dmtp}
(\p_t^2 - \p_{r^*}^2 + f V_{(\ell,m)}^{RW} ) \phi_{(\ell,m)}^- =0 \;\; &\Rightarrow \;\; 
(\p_t^2 - \p_{r^*}^2 + f V_{(\ell,m)}^{Z} )  (\mathcal{D}_{\ell}^+ \phi_{(\ell,m)}^-) =0,\\
(\p_t^2 - \p_{r^*}^2 + f V_{(\ell,m)}^{Z} ) \phi_{(\ell,m)}^+ =0 \;\; &\Rightarrow \;\; 
(\p_t^2 - \p_{r^*}^2 + f V_{(\ell,m)}^{RW} )  (\mathcal{D}_{\ell}^- \phi_{(\ell,m)}^+) =0. \label{dptm}
\end{align}
Chandrasekhar noticed the factorization (\ref{factor}) by casting the RW and Zerilli potentials in Riccati form, and finding 
the unexpected symmetry 
$U_{\ell}^{\pm} = \pm \p_{r^*} W_{\ell} + W_{\ell}{}^2 -{w_{\ell}}^2$. We can  trace origin of this symmetry  to the 
previous to last equation in Section V of \cite{Chaverra:2012bh} , which is equivalent to our equation (\ref{50}) which, 
combined with (\ref{preZ}) gives the $\mathcal{D}_{\ell}^-$ operator in (\ref{dpm})

\subsubsection{Case $\Lambda \geq 0$}

For $\Lambda \geq 0$, we have $-\infty < r^* < \infty$, $U_{\ell}^{\pm} \to 0$ as $|r^*| \to \infty$, and we 
consider $\mathcal{H}^{\pm}_{\ell} $ as an operator in $L^2(\mathbb{R}_{r^*},dr^*)$, where it is self-adjoint 
and positive. Since the general solution of the differential equation $ \mathcal{D}_{\ell}^- \chi=0$ is a constant times 
\begin{equation} \label{chi+}
\chi^+_{\ell} = \frac{r \exp( w_{\ell} \, r^* )}{(\ell+2)(\ell-1)r+6M},
\end{equation}
and the  general solution of the differential equation $ \mathcal{D}^+_{\ell} \chi=0$ is a constant times 
\begin{equation} \label{chim}
\chi^-_{\ell} = \frac{(\ell+2)(\ell-1)r+6M}{r}  \exp(- w_{\ell} \, r^* ) = \frac{1}{\chi^+_{\ell}},
\end{equation}
both $\mathcal{D}^{\pm}$ have  trivial kernel in $L^2(\mathbb{R}_{r^*},dr^*)$. 
Since the evolution of initial data $(\phi_{(\ell,m)}^{\pm},\p_t \phi_{(\ell,m)}^{\pm})|_{t_o}$ in  $L^2(\mathbb{R}_{r^*},dr^*)$ 
gives $\phi_{(\ell,m)}^{\pm}|_{t} \in L^2(\mathbb{R}_{r^*},dr^*)$ for all $t$, we may consider, in view of (\ref{dmtp}) 
(\ref{dptm}), replacing solutions of (\ref{ZE}) with  $\mathcal{D}^-_{\ell}$ times solutions of (\ref{RWE}). 
This was already shown to be possible for $\Lambda \geq 0$ in the previous section, equations  (\ref{z3}) and (\ref{z+}). 
The duality involving the  (\ref{ZE}) and 
 (\ref{RWE}) is now reconsidered from the perspective offered by the factorization (\ref{factor}). 
\begin{lem} \label{lema} 
Assume $\Lambda \geq 0$. For  any solution $\phi_{(\ell,m)}^+$ of  (\ref{ZE}) in $L^2(\mathbb{R}_{r^*},dr^*)$ there 
is a unique solution $\phi_{(\ell,m)}^-$ of  (\ref{RWE}) in $L^2(\mathbb{R}_{r^*},dr^*)$ such that 
\begin{equation} \label{zfrw}
\phi_{(\ell,m)}^+ = \mathcal{D}^+_{\ell} \phi_{(\ell,m)}^-.
\end{equation}
The same statement holds switching $+$ and $ -$ and (RWE) and (ZE).
\end{lem}
\begin{proof}
We will only prove the first statement, as the  proof of the second is completely analogous. Uniqueness follows from 
$\mathcal{D}^-_{\ell}$ being injective. To prove existence, we use the  fact, discovered 
by Price \cite{Price:1971fb}, 
that for $\Lambda=0$, $\phi_{(\ell,m)}^+$  decays  as $t^{-(2\ell+2)}$ at large $t$ (the decay is exponential in $t$ if $\Lambda>0$ \cite{Brady:1996za}) 
The time reversal symmetry of (\ref{ZE}) indicates that 
this also happens for large negative $t$. This implies that 
   $\phi_{(\ell,m)}^+ $  admits a Fourier representation
\begin{equation} \label{fourier}
\phi_{(\ell,m)}^+ = \int_{-\infty}^{\infty}  \widehat \phi_{(\ell,m)}^+(\w,r^*) \; e^{i \w t} \; d\w. 
\end{equation}
Note that 
\begin{equation}
\widetilde \phi_{(\ell,m)}^+ = \int_{-\infty}^{\infty}  \frac{\widehat \phi_{(\ell,m)}^+}{\w^2+w_{\ell}^2} \; e^{i \w t} \; d\w. 
\end{equation}
is also a solution of (\ref{ZE}), therefore $\mathcal{D}^-_{\ell} \widetilde \phi_{(\ell,m)}^+$ is a solution 
of (\ref{RWE}). It is easy to show that this   the solution  sent to $ \phi_{(\ell,m)}^+$ by $\mathcal{D}^+_{\ell}$:
\begin{equation}
 \phi_{(\ell,m)}^+ =  \left(-\p_t^2 +  w_{\ell}^2 \right) \widetilde \phi_{(\ell,m)}^+  = \left( \mathcal{H}^{+}_{\ell} + w_{\ell}^2 \right)
 \widetilde \phi_{(\ell,m)}^+ = \mathcal{D}^+_{\ell} \left( \mathcal{D}^-_{\ell} \widetilde \phi_{(\ell,m)}^+ \right)
\end{equation}
\end{proof}

\subsubsection{Case $\Lambda <0$:  instability of SAdS} \label{ins-neg-lam}

For SAdS, $-\infty<r^*<0$ (equation (\ref{srnl})) and, from (\ref{chi+}),  
 $\chi^+_{\ell} \in L^2(\mathbb{R}^-_{r^*}, dr^*)$ -- the set of square integrable functions on the half line $r^*<0$ -- . 
Since $\mathcal{D}^-_{\ell}\chi^+_{\ell} =0$, we conclude that the operator $\mathcal{D}^-_{\ell}$
 fails to be injective in  $L^2(\mathbb{R}^-_{r^*}, dr^*)$. Also,  
the potential or exponential decay of $\phi_{(\ell,m)}^+$ for large  $|t|$  at fixed position, required  in 
the proof of Lemma \ref{lema}  (equation (\ref{fourier})) allowing to  replace 
solutions of (\ref{ZE}) with solutions of (\ref{RWE}) fails in this case.\\
In fact, using $\mathcal{D}^-_{\ell}\chi^+_{\ell} =0$ together with (\ref{sle}) we get  the 
solution (\ref{ass}) of (\ref{ZE}) found by Chandrasekhar in \cite{ch1}. The choice of a plus sign in (\ref{ass}) 
gives 
\begin{equation} \label{um}
\phi^{+ \;unst}_{(\ell,m), \Lambda<0} = \chi^+_{\ell}(r) \exp(w_{\ell} t) =  \frac{r \exp( w_{\ell} \, (r^*+ t) )}{(\ell+2)(\ell-1)r+6M}, 
\end{equation}
which grows exponentially in time while remaining  in $L^2(\mathbb{R}^-_{r^*}, dr^*)$ for every $t$, that is, is an unstable mode.
 In the  RW gauge and using the coordinates 
$(v,r,\theta,\phi)$  in (\ref{sads}), the metric perturbation from (\ref{um}) 
is given by a particularly simple expression (note that this is is well behaved across the horizon):
\begin{equation} \label{unsh}
h^{+ \;unst}_{(\ell,m), \Lambda<0} =  \exp({w_{\ell} v}) \; S_{(\ell,m)} \left[\frac{w_{\ell}}{6M} \left( r \ell(\ell+1)-6M \right) \; dv \otimes dv 
+ \frac{\ell (\ell+1)}{6M} r^2 \left(d\theta \otimes d\theta + \sin^2(\theta) \; d\phi \otimes d\phi \right)\right].
\end{equation}
As shown in \cite{Araneda:2015gsa}, generic 
perturbations of a Schwarzschild black hole give a type I spacetime, whereas 
the solution (\ref{um})-(\ref{unsh}) splits one of the repeated principal null directions while preserving the degeneracy 
of the other one, leaving a type II spacetime. \\

As explained in Section \ref{hyp}, different dynamics are possible in SAdS depending on the boundary conditions imposed to 
the fields at the timelike boundary. Consider the Zerilli equation in the orbit space SAdS$/SO(3)$
\begin{equation}
(\p_t^2 + \mathcal{H}^+_{\ell})  \phi_{(\ell,m)}^+  = 0.
\end{equation}
The operator $\mathcal{H}^+_{\ell}$ with domain the  compactly supported  functions 
on the half line $C^{\infty}_o(\mathbb{R}^{-}_{r^*})$ 
is symmetric, and it admits  different self adjoint extensions within $L^2(\mathbb{R}^-_{r^*}, dr^*)$. Heuristically, 
self-adjoint extensions of operators like $\mathcal{H}^+_{\ell}$, which have a potential $U^+_{\ell}$ that 
is regular at the boundary, are found by demanding that when integrating by parts  the boundary terms 
do not spoil the transposition of the operator. As the following calculation shows, 
\begin{align} \begin{split}
(\psi_2, \mathcal{H}^+_{\ell} \psi_1) &= \int_{-\infty}^0 \psi_2 (-\p_{r^*}^2 + U^+_{\ell}) \psi_1 dr^* \\
&= \left[ (\p_{r^*} \psi_2)\;  \psi_1 - \psi_2 \; (\p_{r^*} \psi_1) \right]^{r^*=0}_{r^*=-\infty} + \int_{-\infty}^0  \psi_1 (-\p_{r^*}^2 + U^+_{\ell}) \psi_2 dr^*\\
&=  \left[ (\p_{r^*} \psi_2)  \; \psi_1 - \psi_2 \; (\p_{r^*} \psi_1) \right]^{r^*=0}_{r^*=-\infty} + (\mathcal{H}^+_{\ell} \psi_2,\psi_1),
\end{split}  
\end{align}
this will be the case if we restrict $L^2(\mathbb{R}^-_{r^*}, dr^*)$ to the subspace $L^2_{\a}(\mathbb{R}^-_{r^*}, dr^*)$  of functions satisfying the boundary condition (see the discussion around equation (168) in \cite{Ishibashi:2004wx}, which 
applies to our case)
\begin{equation}\label{bcs}
\p_{r^*} \psi \big|_{r^*=0} = \tan(\alpha)\;  \psi \big|_{r^*=0}, 
\end{equation}
$\a \in [-\pi,\pi]$,  where $\alpha= \pm \pi$ (to be identified) is understood as the Dirichlet boundary condition $\psi \big|_{r^*=0}=0$, 
$\alpha=0$  corresponds  to the Neumann boundary condition ($\p_{r^*}\psi \big|_{r^*=0}=0$) 
and  the remaining cases to Robin boundary conditions. 
We will call ${}^{\a}\mathcal{H}^+_{\ell}$ the self adjoint extension of $\mathcal{H}^+_{\ell}$ to the domain $L^2_{\a}(\mathbb{R}^-_{r^*}, dr^*)$. 
Once a self-adjoint extension $\alpha$ is chosen, the dynamics is given by the curve in $L^2_{\a}(\mathbb{R}^-_{r^*}, dr^*)$
obtained   by solving the equation \cite{Wald:1980jn}
\begin{equation} \label{hse}
\p_t^2 \phi + {}^{\a}\mathcal{H}^+_{\ell} \phi =0, \;\;\; \phi \in L^2_{\a}(\mathbb{R}^-_{r^*}, dr^*),
\end{equation}
with initial conditions
\begin{equation} \label{ic}
\phi \big|_{t_o} = p, \;\;\; \p_t\phi  \big|_{t_o} =q, \;\;\; p, q \in L^2_{\a}(\mathbb{R}^-_{r^*}, dr^*), 
\end{equation}
and this is done  by using the resolution of the identity for ${}^{\a}\mathcal{H}^{+}_{\ell}$, i.e., by expanding 
in generalized eigenfunctions of this operator. 
If $E$ is a negative eigenvalue of ${}^{\a}\mathcal{H}^{+}_{\ell}$, it belongs to the discrete part of 
the spectrum and, 
if $p_E$ and $q_E$ 
are the projections of $p$ and $q$ onto the $E-$eigenspace, the projection $\phi_E$ of $\phi$ in (\ref{hse}) will 
be
\begin{equation}
\phi_E(t) = 
q_E \cosh(\sqrt{-E \; }t) + p_E\;  (-E)^{-1/2} \sinh(\sqrt{-E} \;  t). 
\end{equation}
Thus, there is an instability if the spectrum of ${}^{\a}\mathcal{H}^{+}_{\ell}$ contains negative eigenvalues, 
and whether this happens or not may depend on $\a$,  so in general the issue of stability depends on what self 
adjoint extension ${}^{\a}\mathcal{H}^{+}_{\ell}$ (equivalently, what boundary condition at the timelike boundary) 
we choose to define the dynamics.  As an example, in \cite{Andrade:2015gja} Robin boundary conditions are 
enforced at a finite radius $r_D$ on the $\phi^{\pm}_{(\ell,m)}$, as given in equations (2.32) and (2.42). 
This condition assures that there is a gauge for which the induced perturbed metric at the timelike hypersurface  
$r=r_D$ vanishes. One can check from the expressions in \cite{Andrade:2015gja}, however, that 
in the $r_D \to \infty$ limit (2.32) and (2.42) reduce to a the Dirichlet condition 
$\phi^{\pm}_{(\ell,m)} / \p_{r^*} \phi^{\pm}_{(\ell,m)} =0$. \\
For the unstable mode $\chi^+_{\ell}$ we find that it belongs to ${}^{\a}\mathcal{H}^{+}_{\ell}$ with 
\begin{equation}
\tan(\alpha) = \frac{\p_{r^*} \chi^+_{\ell}}{\chi^+_{\ell}} \bigg|_{r^*=0} = w_{\ell} - \frac{2M\Lambda}{(\ell-1)(\ell+2)}.
\end{equation}
SAdS is unstable under dynamics that allow this boundary condition for some $\ell$. A systematic study of  stability 
as a function of $\a$ is being carried out \cite{bernardo}.\\

Note that the even/odd symmetry valid for $\Lambda \geq 0$ (equations (\ref{dmtp})-(\ref{dptm})) and the related isospectrality 
of $\mathcal{H}^{+}_{\ell}$ and $\mathcal{H}^{-}_{\ell}$  (equations (\ref{iso1})-(\ref{iso2})) are in general broken in SAdS 
for generic boundary conditions \cite{Cardoso:2001bb} \cite{Bakas:2008gz}. \\

Note also the similarity between  SAdS  and 
 the negative mass Schwarzschild solution, for which  we may choose $r^*=0$ at $r=0$ and 
\begin{equation}
r^* = r + 2M \ln \left( 1 - \frac{r}{2M} \right), 
\end{equation}
will grow monotonically with  
 $0 < r^* < \infty$. The similarity with  SAdS is that $r^*$ is also restricted to a half-line. This  allows 
to consider the Chandrasekhar mode (the choice of a minus sign in (\ref{ass})) 
\begin{equation} \label{umnm}
\phi^{+ \;unst}_{(\ell,m), M<0} = \chi^+_{\ell}(r) \exp(-w_{\ell} t) =  \frac{r \exp( w_{\ell} \, (r^*- t) )}{(\ell+2)(\ell-1)r+6M},
\end{equation}
as a possible solution of (\ref{ZE}), as it 
belongs for every $t$ to the relevant space $L^2(\mathbb{R}_{r^*}^+,dr^*)$ 
(note that $w_{\ell}<0$ if $M<0$). This solution  grows exponentially in time, so it signals an instability. 
There is, however,  a key difference 
 between the timelike boundary at $r=0$ of the negative mass Schwarzschild solution and the conformal timelike boundary at 
$r=\infty$ of  SAdS, as  for the former there is a unique 
boundary condition that makes the linear perturbation scheme self-consistent, by no worsening the degree 
of the pole of the unperturbed curvature scalars \cite{Gleiser:2006yz} 
\cite{Dotti:2008ta}. The solution  (\ref{umnm}) of (\ref{ZE}) satisfies precisely this boundary condition.
 This was used  in \cite{Gleiser:2006yz} 
\cite{Dotti:2008ta} to 
prove the instability of the Schwarzschild naked singularity.\\

\subsection{Measurable effects of the perturbation on the geometry}

In this section we show that there is a gauge invariant 
combination $G_+$ of  first order variation  of the CSs (\ref{curvaturescalars}) 
 that 
contains all the gauge invariant information about the metric perturbation class 
$[h_{\a \b}^{(+)}]$. In particular, 
$h_{\a \b}^{(+)}$ in the RW gauge can be obtained from $G_+$. 
Following \cite{Dotti:2013uxa}, we introduce the differential curvature scalars (\ref{curvaturescalars}).
The background value of $X$ in the S(A)dS geometry is 
\begin{equation}
X_{S(A)dS} = \frac{M^2}{r^9} (r-2M) - \frac{\Lambda M^2}{3 r^6}, 
\end{equation}
and the background value of $Q_+$ in (\ref{curvaturescalars}) is 
\begin{equation}
{Q_{+}}_{S(A)dS} = \frac{M^2}{r^6}.
\end{equation}
It follows that the combination 
\begin{equation} \label{Gp}
G_+ = (9M-4r+\Lambda r^3) \dot Q_+ + 3 r^3 \dot X
\end{equation}
is gauge invariant since, under a gauge transformation along $\zeta^{\a}$,
\begin{align}\begin{split}
G_+ \to &G_+ +  (9M-4r+\Lambda r^3)\; \pounds_{\zeta}  {Q_+}_{S(A)dS}   + 3 r^3 \; \pounds_{\zeta}  X_{S(A)dS}  \\
= &  G_+ +  (9M-4r+\Lambda r^3) \;  \zeta^r \p_r  {Q_+}_{S(A)dS}   + 3 r^3  \;  \zeta^r \p_r  X_{S(A)dS}  \\
=& G_+ \end{split}
\end{align}
A lengthy calculation with the help of symbolic manipulation programs gives $G_+$ for an arbitrary perturbation 
class 
$[h_{\a \b}^{(+)}]$ using the parametrization $L_+^{\phi}$ in Lemma \ref{evensol}.v and Schwarzschild coordinates $(t,r)$:
\begin{equation} \label{G+}
G_+ = - \frac{2M \dot M}{r^5} + \frac{M}{2r^4} \sum_{\ell \geq 2} \frac{(\ell+2)!}{(\ell-2)!} \left[ f \p_r + Z_{\ell} \right] \phi^+_{(\ell,m)} S_{(\ell,m)},
\end{equation}
where 
\begin{equation} 
Z_{\ell}  = \frac{2M \Lambda r^3+\mu r (r-3M)-6M^2}{r^2 (\mu r + 6M)}, \;\; \mu = (\ell-1)(\ell+2).
\end{equation}
This generalizes the result  equation (33) in \cite{Dotti:2013uxa} to the case $\Lambda \neq 0$.
\begin{thm}
Let $[h_{\a \b}^{(+)}] \in \mathcal{ L}_+$ and  $G_+\left([h_{\a \b}^{(+)}]\right)$ be the  field (\ref{G+}) 
for the perturbation class $[h_{\a \b}^{(+)}]$. 
 The map $[h_{\a \b}^{(+)}] \to G_+\left([h_{\a \b}^{(+)}]\right)$ is invertible: it is possible to construct 
a  representative of $[h_{\a \b}^{(+)}]$ from $G_+\left([h_{\a \b}^{(+)}]\right)$. 
\end{thm}
\begin{proof}
We will prove that the linear map $[h_{\a \b}^{(+)}] \to G_+\left([h_{\a \b}^{(+)}]\right)$ has trivial kernel. 
Assume that $(\dot M_{(1)}, \{ {\phi^+_{(\ell,m)}}^{(1)}\})$ and $(\dot M_{(2)}, \{ {\phi^+_{(\ell,m)}}^{(2)}\})$ give 
the same  $G_+$, then expanding $G_+$ in spherical harmonics we find that $\dot M_{(1)}=\dot M_{(2)}$ and 
also 
\begin{equation} \label{psi1m2}
\left[ f \p_r + Z_{\ell} \right] \chi_{(\ell,m)} =0, \;\;\;\; ( \chi_{(\ell,m)} =  {\phi^+_{(\ell,m)}}^{(1)} -  {\phi^+_{(\ell,m)}}^{(2)}).
\end{equation} 
The general solution of equation (\ref{psi1m2}) is  
\begin{equation} \label{gs}
\chi_{(\ell,m)}(t,r) =  \frac{ F^{(\ell,m)}(t)}{(6M + \mu r) }\; \sqrt{1-\tfrac{2M}{r}-\tfrac{\Lambda r^3}{3}}.
\end{equation}
Since $\chi_{(\ell,m)}(t,r) =  {\phi^+_{(\ell,m)}}^{(1)} -  {\phi^+_{(\ell,m)}}^{(2)}$, it must satisfy Zerilli's equation.
Inserting   (\ref{gs}) in (\ref{ZE}) gives 
\begin{equation}
\frac{d^2F^{(\ell,m)}}{dt^2}  - \left( \frac{(\mu+1) \Lambda}{3}
- \frac{\mu}{r^2}+ \frac{2M(\mu -2)}{r^3}+ \frac{9M^2}{r^4} \right) F^{(\ell,m)} =0, \;\;\;\; (\mu = (L-1)(L+2)),
\end{equation}
which only admits the trivial solution $F(t)=0$.\\

To construct a representative of $[h_{\a \b}^+]$ we need $\dot M$ and the $\phi^+_{(\ell,m)}$, together with equations (\ref{pert+T}), (\ref{ttos}), 
(\ref{md}), (\ref{pert-RW}), (\ref{preZ}), (\ref{zerilli-exp}), (\ref{J}) and (\ref{H}).  Expanding $G_+$ in spherical harmonics
\begin{equation}
G_+ = \sum_{(\ell,m)} G_+^{(\ell,m)} S_{(\ell,m)}
\end{equation}
we find that $\dot M = -r^5 G_+^{(\ell=0)}/(2M)$ and $\phi^+_{(\ell,m)}$ is the only solution of
\begin{equation}
 \left[ f \p_r + Z_{\ell} \right] \phi^+_{(\ell,m)}  = \frac{2r^4}{M} \frac{(\ell-2)!}{(\ell+2)!} G_+^{(\ell,m)}
\end{equation}
that satisfies 
(\ref{ZE}).

\end{proof}

\subsection{Non-modal linear stability of the $\Lambda \geq 0$ black holes}

In this section we establish the pointwise boundedness of $G_+$, equation (\ref{G+}),  on region II of a Schwarschild  or SdS 
black hole. For $\Lambda<0$ and certain boundary conditions at the timelike boundary, this fails 
to be true. As an example, for boundary conditions allowing the unstable mode (\ref{um}), $G_+$ 
contains a contribution proportional to 
\begin{equation}
G_+\left( [ h^{+ \;unst}_{(\ell,m)} ] \right) = 
G^{+ \;unst}_{(\ell,m)} =  \frac{M}{2r^4}  \frac{(\ell+2)!}{(\ell-2)!} \left( \frac{\ell(\ell+1)}{12M} - \frac{1}{2r} \right) 
\exp({w_{\ell} v}) \; S_{(\ell,m)}.
\end{equation}
where $h^{+ \;unst}_{(\ell,m)}$ is given in (\ref{unsh}). \\

Back to the $\Lambda \geq 0$ case, we know 
from Lemma \ref{lema} that for a  given a  solution $\phi^+_{(\ell,m)}$ of the (\ref{ZE}), there exists a solution $\phi^-_{(\ell,m)}$ 
of the (\ref{RWE}) such that 
  $\phi^+_{(\ell,m)}=\mathcal{D}^+ \phi^-_{(\ell,m)}$. Using this replacement  in (\ref{G+}) gives 
\begin{multline} \label{replacement}
  \left[ f \p_r + Z_{\ell} \right] \phi^+_{(\ell,m)}  =  \left[ f \p_r + Z_{\ell} \right] \left[ f \p_r + W_{\ell} \right] \phi^-_{(\ell,m)}\\
= \p_{r^*}^2 \phi^-_{(\ell,m)}   + (\p_{r^*} W_{\ell} + Z_{\ell} W_{\ell} )  \phi^-_{(\ell,m)}  + ( W_{\ell} + Z_{\ell} ) \p_{r^*} \phi^-_{(\ell,m)}\\
= \p_{t}^2 \phi^-_{(\ell,m)}   + (f V^{RW}_{\ell} + \p_{r^*} W_{\ell} + Z_{\ell} W_{\ell} )  \phi^-_{(\ell,m)}  + ( W_{\ell} + Z_{\ell} ) \p_{r^*} \phi^-_{(\ell,m)}.
\end{multline}
Since
\begin{equation} \label{cancel}
f V^{RW}_{\ell} + \p_{r^*} W_{\ell} + Z_{\ell} W_{\ell}  = -\tfrac{1}{6} \Lambda \ell (\ell+1) + \frac{\Lambda M + w_{\ell}}{r} 
+ \frac{\ell(\ell+1)-6M w_{\ell} }{2r^2}- \frac{(\ell (\ell+1)+3)M}{r^3} + \frac{6M^2}{r^4} =: \sum_{j=0}^4 P_j(\ell) r^{-j}
\end{equation}
and 
\begin{equation}
W_{\ell} + Z_{\ell} = w_{\ell} + \frac{r-3M}{r^2},
\end{equation}
we can re write (\ref{G+}) as 
\begin{equation} \label{G++}
G_+ = - \frac{2M \dot M}{r^5} + \frac{M}{2r^3} \left[ \p_t^2 \Phi_5 + \sum_{j=0}^4 r^{-j} \Phi_j + \frac{f (r-3M)}{r^3} \Phi_5 
+ \frac{f}{r} \Phi_6 \right] + \frac{M}{2 r^3} \p_{r^*} \Phi_6 + \frac{M(r-3M)}{2 r^5} \p_{r^*} \Phi_5,
\end{equation}
where 
\begin{align}
\Phi_j &= \sum_{(\ell\geq2, m)} \frac{(\ell+2)!}{(\ell-2)!} P_j(\ell) \frac{\phi_{(\ell,m)}}{r} S_{(\ell,m)}, \;\; j=0,1,2,3,4\\
\Phi_5 &=  \sum_{(\ell\geq2, m)} \frac{(\ell+2)!}{(\ell-2)!}\frac{\phi_{(\ell,m)}}{r} S_{(\ell,m)},\\
\Phi_6 &=  \sum_{(\ell\geq2, m)} \frac{(\ell+2)!}{(\ell-2)!} w_{\ell}\frac{\phi_{(\ell,m)}}{r} S_{(\ell,m)},
\end{align}
are all solutions of (\ref{4DRWE}), and therefore so is  $\p_t^2 \Phi_5$. Note that the expression (\ref{G++}) for $G_+$,  
written entirely in terms of solutions of  (\ref{4DRWE}) is 
possible thanks to the cancellation  in (\ref{cancel}) of the $((\ell+2)(\ell-1)r+6M)$ denominators in $Z_{\ell}$ and 
$W_{\ell}$.

\begin{thm} \label{even}
 \mbox{}
\begin{itemize}
\item[(i)]  For any smooth solution of the even LEE which has  compact support on Cauchy surfaces of the Kruskal extension
$ \rm{I \cup II \cup I' \cup II'}$ of 
the Schwarzschild space-time, there exists a constant $K_+$ such that $|G_+ | < K_+\; r^{-4}$ for $r>2M$.
\item[(ii)] For any mooth solution of the even LEE which has  compact support on Cauchy surfaces of region $ \rm{I \cup II \cup I' \cup II'}$ of the extended SdS 
black hole, there exists a constant $K_+$ such 
that $|G_+ | < K_+ \; r^{-4}$ (equivalently $|G_+ | <$ some constant) for $r_h<r<r_c$.
\end{itemize}
\end{thm}
\begin{proof}
We treat simultaneously  the $\Lambda=0$ and $\Lambda>0$ cases. 
For even solutions of the LEE with   compact support on Cauchy surfaces of the Kruskal extension,  there is an open 
neighbourhood $\mathcal{N}$
of  $r^*=\infty$  where $\phi_{(\ell,m)}^+=0$. This set is of the form  
$\mathcal{N}= \{ (t,r^*,\theta,\phi) \; | \; r^* > R^*(t) \}$ where, for  large $t$, 
$R^*(t) = t+$ constant, and there is a similar 
open neighbourhood $\mathcal{N}'$ near the $\mathcal{I}^+$ boundary 
 of 
the isometric region II$'$. \\
In (\ref{replacement}), we replaced the $\phi_{(\ell,m)}^+$ fields with $\phi_{(\ell,m)}^-$ fields according to    (\ref{zfrw}). 
Using    $\phi_{(\ell,m)}^+\large|_{\mathcal{N}}=0$, equations (\ref{dpm})-(\ref{dpm2}), 
 (\ref{factor}) and  (\ref{chim}),
 and the fact that $\phi_{(\ell,m)}^-$ 
 satisfies (\ref{RWE}), we find that,  in $\mathcal{N}$,
\begin{equation} \label{2}
\phi_{(\ell,m)}^- = \frac{(\ell+2)(\ell-1)r+6M}{r}  \exp(- w_{\ell} \, r^* ) 
 [A_{(\ell,m)} \exp(- w_{\ell} \, t) + B_{(\ell,m)} \exp(w_{\ell} \,t)].
\end{equation}

According to Theorems \ref{odd-0} and \ref{odd+}, solutions $\Phi$ of the 4DRWE  with compact support on Cauchy slices of the  Kruskal extension
 satisfy $|\Phi|<C/r$ in region II. The  
 4DRW fields $\Phi_j, j=0,...,6$ in equation (\ref{G++}) do not have compact support on Cauchy slices, 
 their  spherical harmonic 
components have the exponential tails 
 (\ref{2}). In what follows we prove that the results in  Theorems \ref{odd-0} and \ref{odd+} hold also in this case. This 
will allow  us to place a pointwise bound to the term between square brackets in  $G_+$, equation  (\ref{G++}). \\
Let $\Phi$ be any of the $\Phi_j, j=0,...,6$,  and write $\Phi$ as a sum of three solutions of the 4DRWE, $\Phi=\Phi_{(a)}+\Phi_{(b)}+\Phi_{(c)}$, 
 where $\Phi_{(a)}$ has compact support on Cauchy slices of the  Kruskal extension and $\Phi_{(a)}=0$ 
for $r^* > R^*(t)+\epsilon$ for some positive $\epsilon$,  
 $\Phi_{(b)}$ is supported in $r^* \in (R^*(t)-\epsilon, \infty)$  and $\Phi_{(c)}$ 
is similarly supported near the $\mathcal{I}^+$ boundary of region II$'$ (this is done by writing the $\Phi$ datum on a $t$ slice as a sum of three appropriate terms 
and letting them evolve). In region II, $\Phi= \Phi_{(a)}+\Phi_{(b)}$ and $\Phi_{(a)} < C_{(a)}/r$ since 
its satisfies the hypothesis in Theorem \ref{odd-0} (\ref{odd+}). On the other hand, $\Phi_{(b)}$ decays  as in (\ref{2}). Since this field 
  is bounded away the 
bifurcation sphere, we can   adapt the proof of boundedness in the Appendix 
in \cite{Kay:1987ax}. This proof follows two steps: i) the use of the Sobolev-type inequality
\begin{equation} \label{sobo}
\big| r \Phi|_t \big| \leq K \left( || \Phi  || +|| r^{-1} \p_{r^*}^2 (r\Phi) ||
+  || \mathbf{J}^2 \Phi || \right), 
\end{equation}
where the norm 
\begin{equation} \label{norm}
||\Phi || = \langle \Phi | \Phi \rangle^{1/2}
\end{equation}

is given by 
the inner product 
\begin{equation}
\langle \Phi | \Psi \rangle =  \int \Phi \Psi \; r^2 \; dr^* \sin(\theta) \; d \theta d \phi,
\end{equation}
and ii) the replacement of  each of the terms on the right hand side of (\ref{sobo}) with $t-$independent 
quantities, which gives $|\Phi|<$constant$/r$.  Since equation (\ref{sobo}) holds for finite norm fields,  we only need to 
 prove that step ii) is feasible for  $\Phi_{(b)}$. We will prove this avoiding   the use of the 
operators  $\mathcal{A}^{-1/2}$ and $ (\mathcal{A}^{\Lambda})^{-1/2}$) introduced in \cite{Kay:1987ax}, 
applying  instead  the equivalent time integral technique  introduced in sections 3.4 and 3.5 of \cite{Dain:2012qw} 
for the  similar problem of proving  the pointwise boundedness of $\Phi_{KG}$ on the 
static region of an extreme Reissner-Nordstr\"om black hole, $\Phi_{KG}$ a solution of the Klein-Gordon equation
 (note that $s$ is used in \cite{Dain:2012qw} for the tortoise 
coordinate $r^*$).   The idea is the following: 
for solutions $\Phi \in L^2(S^2)_{>1}$ of the 4DRWE 
with finite norm 
 there is a   conserved (i.e., time independent) energy given by
\begin{equation} \label{energy}
 \mathcal{E}[\Phi] =  \langle \p_t \Phi | \p_t \Phi   \rangle + 
\langle  \Phi | r^{-1} \mathcal{A}^{\Lambda} (r \Phi)   \rangle.
\end{equation}
This is easily checked  using the form  (\ref{2+2}) of the 4DRWE. Note that both terms contributing to $ \mathcal{E}[\Phi]$ are positive
for a finite norm $\Phi \in L^2(S^2)_{>1}$,  since $V_1-V_2 \hd^A \hd_A>0$ as a consequence of  (\ref{rhrc}) and (\ref{V12a}) and 
$\p_{r^*}^2>0$. 
Now assume there is a finite energy time integral $\tilde \Phi$ of $\Phi$, that is,  $\tilde \Phi$ satisfies the 4DRWE and 
$\p_t \tilde \Phi = \Phi$. Then $ \mathbf{J}^2 \tilde \Phi$ also satisfies the 4DRWE and
\begin{align} \label{i1}
|| \Phi ||^2 &\leq  \mathcal{E}[\tilde \Phi] \\ \label{i2}
|| r^{-1} \p_{r^*}^2 (r\Phi) ||^2 &\leq  \mathcal{E}[\p_t \Phi] \\
|| \mathbf{J}^2 \Phi ||^2 &\leq  \mathcal{E}[\mathbf{J}^2 \tilde \Phi]. \label{i3}
\end{align}
 (compare with the set of equations 
above (A2) in \cite{Kay:1987ax}, and with equation (76) in \cite{Dain:2012qw}). 
This allows us to replace the right hand side of equation (\ref{sobo}) with a time independent constant, and 
get the desired bound $\Phi < C/r$. 
Note that (\ref{i1}) and (\ref{i3}) follow straightforwardly from the definition (\ref{energy}) 
and the comments below it, whereas (\ref{i2}) follows from $\mathcal{E}[\p_t \Phi] \geq ||\p_t^2 \Phi ||^2 = ||  r^{-1} \mathcal{A}^{\Lambda} (r \Phi) ||^2
> || r^{-1} \p_{r^*}^2 (r\Phi) ||^2 $. 
The existence of  time integrals $\tilde \Phi_{(a)}$ of $\Phi_{(a)}$ and 
 $\tilde \Phi_{(b)}$ of $\Phi_{(b)}$ 
can be proved following the steps in section 3.4 of \cite{Dain:2012qw}, since these fields  belong to  the Hilbert space 
$\hat{\mathcal{H}}$ introduced in Lemma 2 of this reference (with the appropriate replacements of $V_1$ and $V_2$). 
To prove  that $\tilde \Phi_{(a)}$ has finite energy we proceed as in  Lemma 3 in  \cite{Dain:2012qw}, 
since $\Phi_{(a)}$ supported away of $r^*=\infty$. 
To prove that 
$\tilde \Phi_{(b)}$ has finite energy,  we use 
   (\ref{2}) to show that, in $\mathcal{N}$, 
\begin{equation} \label{22}
{ \tilde \phi_{(b)}}{}_{(\ell,m)} = \frac{(\ell+2)(\ell-1)r+6M}{w_{\ell} \; r}  \exp(- w_{\ell} \, r^* ) \;
 [-A_{(\ell,m)} \exp(- w_{\ell} \, t) + B_{(\ell,m)} \exp(w_{\ell} \,t)] + \psi^{-,o}_{(\ell,m)}(r) 
\end{equation}
where $\psi^{-,o}_{(\ell,m)}(r)$ is a zero mode ($\omega=0$) solution of equation (\ref{sle}), that is 
$\mathcal{A} \psi^{-,o}_{(\ell,m)}=0$ 
($\mathcal{A}^{\Lambda} \psi^{-,o}_{(\ell,m)}=0$). The  asymptotic form for large $r^*$ for a zero mode is
\begin{equation} \label{23}
\psi^{-,o}_{(\ell,m)} \sim \begin{cases} A_{(\ell,m)} [{r^*}^{-\ell} + ...] +
 B_{(\ell,m)} [ {r^*}^{\ell+1}+...] & ,\Lambda=0 \\
A_{(\ell,m)} [e^{-\alpha r^*} + ...] + B_{(\ell,m)} [ 1+...] & ,\Lambda >0, 
\end{cases}
\end{equation}
where $\alpha >0$ is defined  by equation  (\ref{rs+L}) together with  $e^{-\alpha r^*} = (r-r_c)/r_c$. 
For the unique  time integral in $\hat{\mathcal{H}}$ of Lemma 2 in \cite{Dain:2012qw} it must be  $B_{(\ell,m)}=0$ in (\ref{23}). 
This implies  that the energy integrals (\ref{i1}) and (\ref{i3}) converge for $\Phi=\tilde \Phi_{(b)}$.
 We conclude that $\Phi_{(b)} < C_{(b)}/r$ and that the sum of terms within square brackets in (\ref{G++}) 
is bounded by $C'/r$ for some positive $C'$.  \\

To deal with the last two terms in (\ref{G++}) we need to 
prove the pointwise boundedness of $\p_{r^*}\Phi$ in region II for $\Phi$ a solution of the 4DRW, 
$\Phi= \Phi_{(a)}+\Phi_{(b)}+\Phi_{(c)}$ as above. We do so  by 
adapting  the proof in section 3.6 in \cite{Dain:2012qw}  of the pointwise boundedness of $\p_{r^*} (r\Phi_{KG})$.  
For $\Phi=\Phi_{(a)}$ or $\Phi=\Phi_{(b)}$ 
start from the Sobolev inequality (\ref{sobo}) with $\Phi$ replaced with $r^{-1} \p_{r^*}(r \Phi)$:
\begin{equation} \label{sobotrol}
\big| \p_{r^*}(r \Phi)|_t \big| \leq K \left( || r^{-1}  \p_{r^*}(r \Phi) || +
|| r^{-1} \p_{r^*}^3 (r\Phi) || 
+  || \mathbf{J}^2  r^{-1}\p_{r^*}(r \Phi)  || \right) 
\end{equation}
The square of the  first term on the right hand side is bounded by the time independent energy of $\Phi$:
\begin{align} \nonumber
 \mathcal{E}[\Phi] &=  \langle \p_t \Phi | \p_t \Phi   \rangle + 
\langle  \Phi | r^{-1} \mathcal{A}^{\Lambda} (r \Phi)   \rangle \\  \nonumber
&\geq \langle  \Phi | r^{-1} \mathcal{A}^{\Lambda} (r \Phi)   \rangle \\  \nonumber
&\geq \langle  \Phi | -r^{-1} \p_{r^*}^2 (r \Phi)   \rangle \\ 
&= || r^{-1}\p_{r^*}(r \Phi)||^2. 
\end{align}
The second inequality above follows from  $V_1-V_2 \hd^A \hd_A$ being positive definite on 
fields in $L^2(S^2)_{>1}$,  the integration 
by 
parts in the last line is trivial  for the compactly supported $\Phi_{(a)}$ and,  in view of   (\ref{22})-(\ref{23}),  holds  for 
$\Phi_{(b)}$. The third term in (\ref{sobotrol}) is similarly bounded 
by $ \mathcal{E}[ \mathbf{J}^2 \Phi]$. To treat the second term on the right hand side of 
(\ref{sobotrol}) 
we proceed as in \cite{Dain:2012qw},
by taking  the $r^*$ derivative of $\mathcal{A}^{\Lambda} (r \Phi)$  (see  (\ref{AL}))
\begin{equation}
 \p_{r^*}^3 (r\Phi) = - \p_{r^*} \mathcal{A}^{\Lambda} (r \Phi) + r\Phi \p_{r^*} V_1 +V_1 
\p_{r^*}  (r\Phi)  - (\p_{r^*} V_2)  \mathbf{J}^2  (r\Phi) - V_2  \mathbf{J}^2  \p_{r^*}  (r\Phi)
\end{equation}
and using  the facts  that $\Phi'=r^{-1} \mathcal{A}^{\Lambda} (r \Phi)$ is also a solution 
of the 4DRWE and that the $V_i$'s and their $r^*$ derivatives are bounded in region II: $|V_i| \leq V_i^{max}$, 
$|\p_{r^*} V_i| \leq V_{i,r^*}^{max}$:
\begin{multline}
||r^{-1} \p_{r^*}^3 (r\Phi) || \leq ||r^{-1} \p_{r^*} \mathcal{A}^{\Lambda} (r \Phi) || + || \Phi  || V_{1,r^*}^{max} 
+V_1^{max} 
||r^{-1} \p_{r^*}  (r\Phi)||  + V_{2,r^*}^{max} || r^{-1} \mathbf{J}^2  (r\Phi) || + V_2^{max}
|| r^{-1}  \mathbf{J}^2  \p_{r^*}  (r\Phi) || \\
\leq \sqrt{\mathcal{E}[r^{-1}  \mathcal{A}^{\Lambda} (r \Phi)]} + \sqrt{\mathcal{E}[\tilde \Phi]} \;  V_{1,r^*}^{max} 
+ V_1^{max}  \sqrt{\mathcal{E}[\Phi]} 
 + V_{2,r^*}^{max} \sqrt{\mathcal{E}[r^{-1}  \mathbf{J}^2  (r \tilde \Phi)]}
 + V_{2}^{max} \sqrt{\mathcal{E}[r^{-1}  \mathbf{J}^2  (r\Phi) ]}
\end{multline}
The finiteness of the energy integrals above for $\Phi=\Phi_{(b)}$ can easily be checked. We conclude that 
there is a constant $K'$ such that, in region II, 
\begin{equation}
K' \geq \big| \p_{r^*}(r \Phi) \big| \geq f |\Phi| + r |\p_{r^*} \Phi| 
\end{equation}
and, since $f$ and $\Phi$ are bounded in this region,
\begin{equation}
|\p_{r^*} \Phi|  \leq \frac{K''}{r}
\end{equation}
for some constant $K''$. 
Thus, every term in (\ref{G++}) is bounded by a constant times $r^{-4}$.
\end{proof}

\section{Discussion}

\subsection{Evolution  of perturbations}

The large $t$ decay of $\phi^{\pm}_{(\ell,m)}$ \cite{Price:1971fb,Brady:1996za},  together with equations  
\begin{equation*} \tag{\ref{qdt}}
G_- = - \frac{6M}{r^7} \sqrt{\tfrac{4 \pi}{3}}  \sum_{m=1}^3 j^{(m)}S_{(1,m)} - \frac{3M}{r^5} 
\sum_{(\ell\geq2, m)} \frac{(\ell+2)!}{(\ell-2)!} \frac{\phi^-_{(\ell,m)}}{r} 
S_{(\ell,m)},
\end{equation*}
 and
\begin{equation*} \tag{\ref{G+}}
G_+ = - \frac{2M \dot M}{r^5} + \frac{M}{2r^4} \sum_{\ell \geq 2} \frac{(\ell+2)!}{(\ell-2)!} \left[ f \p_r + Z_{\ell} \right] \phi^+_{(\ell,m)} S_{(\ell,m)},
\end{equation*}
 indicate that, at large $t$,
\begin{equation}
G_- \sim - \frac{6M}{r^7} \sqrt{\tfrac{4 \pi}{3}}  \sum_{m=1}^3 j^{(m)}S_{(1,m)}, \;\; \; G_+ \sim - \frac{2M \dot M}{r^5},
\end{equation}
which, in view of the bijection (\ref{im}), corresponds to a linearized Kerr (Kerr de Sitter) black hole around the Schwarzschild (SdS) 
background. 
The picture that emerges from these considerations is that, for a generic perturbation, the black hole ends up settling into a slowly 
rotating Kerr (Kerr dS) black hole. To make statements like these more precise, all we need  is quantitative  
information on the decay of solutions of the 4DRW equation. This is so because $r^5 G_-$ {\em is} a solution of this equation, 
and  $G_+$ can  be written entirely in terms of solutions of the 4DRW equation (see equation (\ref{G++})). 
Alternatively, in view of Theorem \ref{bij}, there is a bijection between solutions of the LEE and 
a set containing the constants $\dot M,  j^{(m)}$ and two fields $\Phi_{\pm}$ that obey  the 4DRW equation. 
The perturbed metric is linearized Kerr (Kerr de Sitter) 
if $\Phi_{\pm}=0$, which is the limit approached at large $t$. 
Quantitative results for the decay in time  of solutions of the 4DRW equation can be found in \cite{Blue:2003si}, \cite{fm} and more recently in  
\cite{Dafermos:2016uzj} (in this last reference, as statements on a symmetric tensor field made out of two scalar 
fields obeying  (\ref{4DRWE})). \\

\subsection{Reduction of the LEE to  the 4DRW equation}

A natural question that arises from our results  is:
can we forget altogether  metric perturbations and restrict ourselves to the study of the  4DRW equation?
The answer to this question is in the affirmative (Theorem \ref{bij} and equations (\ref{qdt}) and (\ref{G++})), although  some technical issues  should be mentioned 
if one intends to study perturbations using only the $G_{\pm}$ or similar perturbed curvature fields.  
One is related to the possibility of evolving perturbations from initial $G_{\pm}$ data. 
Since $r^5 G_-$ obeys the 4DRW equation, initial data (e.g., $\p_t G_-$ and $G_-$ restricted to a $t=$constant hypersurface) 
gives  $G_-$ and, by expanding in spherical harmonics, the metric perturbation in the RW gauge  uniquely related to it. 
This does not happen, however, with $G_+$ in the even sector. 
Since the operator $f \p_r + Z_{\ell}$ in (\ref{G+}) has a non trivial kernel, and $G_+$  does not satisfy a wave equation 
(which would allow us to recover the lost information by proceeding as 
in Section \ref{chd}), the set of $\p_t \phi^+_{(\ell,m)}$ and $\phi^+_{(\ell,m)}$ at 
$t=t_o$ cannot be obtained from  $\p_t G_+$ and $G_+$ at $t=t_o$. The reason why $G_+$ does not obey 
a wave equation is, of course, that is made out of the perturbation of curvature scalars 
involving  both the Riemann tensor {\em and its covariant derivative};  
the possibility of constructing a gauge invariant curvature scalar  field for the even sector 
that does not use derivatives of the Riemann tensor was ruled out in \cite{Dotti:2013uxa}. \\
We should mention, however,  two alternatives to the use of $G_+$ (or  similar fields involving CSs).  
One is using the potentials $\Phi_{\pm}$ in Theorem \ref{bij}: 
perturbations are entirely characterized by  $\dot M,  j^{(m)}$ (stationary perturbations) and the  fields  $\Phi_{\pm}$
obeying the 4DRW equation
 (dynamical perturbations). This shows that we can do without the Zerilli equation and that 
 the LEE does reduce to the  4DRW equation. 
The other possibility to avoid scalar fields involving higher metric derivatives (such as $\dot X$ in $G_+$, equation (\ref{Gp})), is constructing gauge invariant 
combinations made out of   perturbed curvature 
scalars {\em and} the  metric perturbation. We close this section by exhibiting  an example of 
such a construction: if we calculate 
the (gauge dependent) scalar field $\dot Q_+$ (see (\ref{curvaturescalars})) in the RW gauge, we find that 
\begin{equation} 
\dot Q_+^{(RW)} = \frac{2M \dot M}{r^6} - \frac{6M^2}{r^5} \sum_{(\ell\geq 2,m)} \frac{\mathcal{D}^-_{\ell} \phi^+_{(\ell,m)}}{r} S_{(\ell,m)}, 
\end{equation}
thus $r^5  \dot Q_{+,>1}^{(RW)}$ also satisfies the 4DRW equation! We might think of using  $\dot Q_+$ 
to measure the strength of even perturbations, but 
this field 
 is gauge dependent  due to the fact that  $Q_+ \neq 0$ for the background Schwarzschild or S(A)dS black hole. 
Under the gauge transformation (\ref{gt}),
\begin{equation} \label{qp1}
\dot Q_+ \to {\dot Q_+}{}' = \dot Q_+ +  \zeta^r \p_r Q_+ = \dot Q_+ -  \zeta^r  \; \frac{6M^2}{r^7}.
 \end{equation}
This suggests  searching for  a gauge invariant field $H_+$  that agrees with $\dot Q_+$ in the RW gauge. 
If we compare (\ref{qp1}) with equation (\ref{gauge+}), we find the following solution: 
\begin{equation}
H_+ =   \frac{6M^2}{r^7} \td^a r \left( q_a -r^2 \td_a G \right) + \dot Q_+
\end{equation}
Since this gauge invariant field {\em  reduces in the RW gauge to} $\dot Q_+^{(RW)}$, and the $\phi^+_{(\ell,m)}$ are 
gauge invariant, 
\begin{equation} 
H_+ = \frac{2M \dot M}{r^6} - \frac{6M^2}{r^5} \sum_{(\ell\geq 2,m)} \frac{\mathcal{D}^- \phi^+_{(\ell,m)}}{r} S_{(\ell,m)}.
\end{equation}
Note that the dynamical $\ell \geq 2$ piece of $r^5 H_+$ satisfies the 4DRW equation, which  gives $H_+$ an advantage over $G_+$: it is possible 
 to obtain $H_+$  from the  initial datum $(r^5 H_+, r^5 \p_t H_+)|_{t_o}$. Once this is done, the corresponding metric 
perturbation in the RW gauge can be obtained by expanding $H_+$ in spherical harmonics, which gives us $\dot M$ and 
the $\mathcal{D}^-_{\ell} \phi^+_{(\ell,m)}$ and then applying Lemma \ref{lema} to recover the Zerilli fields $\phi^+_{(\ell,m)}$.
Although the $H_+$ looks more geometrical than the potentials $\Phi_{\pm}$ in Theorem \ref{bij}, we have not found 
an obvious interpretation for this field.

\section{Acknowledgments}
I thank Sergio Dain for suggesting writing a  detailed version of the proof of nonmodal linear stability in 
\cite{Dotti:2013uxa},  Andr\'es Anabal\'on for suggesting considering extending this proof to the  case of  nonzero cosmological 
constant, Gustav Holzegel for pointing out an error in the proof of Theorem \ref{even} in a previous version of this manuscript 
and 
Reinaldo Gleiser and Martin Reiris for useful  discussions. 
This work was partially funded by grants PIP 11220080102479
(Conicet-Argentina) and Secyt-UNC 05/B498 (Universidad Nacional de C\'ordoba).

\end{document}